\documentclass[a4paper,12pt]{article}

\usepackage[danish,english,]{babel}
\usepackage[T1]{fontenc}
\usepackage[utf8]{inputenc}
\usepackage{lmodern}

\usepackage{enumitem}
\usepackage{amsthm,amsfonts}
\usepackage{mathtools}
\usepackage{natbib}
\usepackage{bbm}
\usepackage{prodint}
\usepackage{siunitx}
\usepackage{booktabs}
\usepackage{changepage}

\theoremstyle{definition}
\newtheorem{observation}{Observation}
\newtheorem{assumption}{Assumption}
\theoremstyle{plain}
\newtheorem{theorem}{Theorem}
\newtheorem{proposition}{Proposition}
\newtheorem{corollary}{Corollary}
\newtheorem{lemma}{Lemma}

\newcommand{\R}{\mathbb{R}}
\renewcommand{\P}{\operatorname{P}}
\newcommand{\E}{\operatorname{E}}
\newcommand{\Var}{\operatorname{Var}}
\newcommand{\Cov}{\operatorname{Cov}}
\newcommand{\mad}{\operatorname{MAD}}
\newcommand{\given}{\mathbin{\mid}}
\renewcommand{\d}{\mathrm{d} \,}
\newcommand{\T}{^{\mathsf{T}}}
\newcommand{\indic}[1]{\mathbbm{1}(#1)}
\newcommand\independent{\protect\mathpalette{\protect\independenT}{\perp}}
\def\independenT#1#2{\mathrel{\rlap{$#1#2$}\mkern2mu{#1#2}}}

\title{A comparison of Kaplan--Meier-based inverse probability of censoring weighted regression methods}
\author{Morten Overgaard}
\date{\footnotesize Department of Public Health, Aarhus University, Denmark \\
E-mail: moov@ph.au.dk}

\begin{document}

\maketitle

\begin{abstract}
Weighting with the inverse probability of censoring is an approach to deal with censoring in regression analyses where the outcome may be missing due to right-censoring. In this paper, three separate approaches involving this idea in a setting where the Kaplan--Meier estimator is used for estimating the censoring probability are compared. 
In more detail, the three approaches involve weighted regression, regression with a weighted outcome, and regression of a jack-knife pseudo-observation based on a weighted estimator. 
Expressions of the asymptotic variances are given in each case and the expressions are compared to each other and to the uncensored case. In terms of low asymptotic variance, a clear winner cannot be found. Which approach will have the lowest asymptotic variance depends on the censoring distribution.
Expressions of the limit of the standard sandwich variance estimator in the three cases are also provided, revealing an overestimation under the implied assumptions.
\end{abstract}

\section{Introduction}

In some settings of survival analysis, the primary interest may concern one or a few time points, where measures of interest could be risk or restricted mean survival time. By focusing on one or a few time points, the analyst can easily communicate results. As pointed out by \cite{Martinussen2023}, a particular time point may have a special role in a specific clinical subject area since disease-free survival to such a time point may be taken to indicate cure of the patient. 
A relevant tool for adjusted comparisons of the measures of interest between groups in this sort of setting is a regression method such as logistic regression.
The issue of right censoring becomes a missing data problem for such regression methods by potentially rendering the outcome unobserved.

This paper is concerned with a regression setting and the handling of outcomes that may be missing due to right censoring. Here, the handling will in some way be by use of inverse probability weighting. The weights are related to the censoring distribution and for this reason such approaches may be considered inverse probability of censoring weighting (IPCW) approaches. 
On the other hand, this may be considered a misnomer since the weights are, much like sampling weights, estimates of the inverse probabilities of observing the outcomes in question rather than probabilities of censoring. 
This paper will consider three approaches where the weights are based on Kaplan--Meier estimates of the censoring distribution, either based on the entire sample or calculated in strata. An important assumption will be independent censoring.

One approach involves a Horvitz--Thompson-type weighting of the entire individual contribution to an estimating equation, which has been considered in a survival setting by e.g.\ \cite{Robins1992}. 
A second approach, which was suggested by \cite{Scheike2008} for assessing the influence of covariates on a cumulative incidence curve, involves weighting only the potentially missing outcome for each individual.
The third approach that is considered here was introduced by \cite{Andersen2003} and involves replacing the potentially missing outcomes by jack-knife pseudo-values, here based on an inverse probability weighting estimator such as the Kaplan--Meier estimator. 

The paper by \cite{blanche2022logistic} studies and compares the first and second approach in a setting where the outcome is having a certain event within a certain time and where a logistic regression is considered. A main conclusion of that paper is that which method is more efficient depends on the censoring distribution. 
Another conclusion is that a naive approach to variance estimation will lead to too large, or conservative, variance estimates. 

The third approach, also known as the pseudo-observation method, is reviewed by \cite{Andersen2010} where some further background can be found.
Much of the theory on the method that is useful for the purposes of this paper can be found in \cite{Overgaard2019}. 
The paper by \cite{Andersen2010} also suggests stratifying calculation on a variable if censoring depends on that variable, which is an approach studied in more detail in this paper.

Although the theory of these three different approaches seems clear, at least in some settings, it is, however, not clear how the three approaches compare, especially how the pseudo-observation method compares with the two other approaches in terms of efficiency. 
Results by \cite{Binder2014} and \cite{parner2023regression} indicate that the pseudo-observation method makes more efficient use of data compared to the outcome weighting approach,
but it remains unclear why this is the case and how generally this holds.

In this paper, a comparison of the three approaches is carried out theoretically, in a theoretical example, and in simulations. The theoretical comparison is in terms of the asymptotics of the three approaches in a general setting which is laid out within a common framework that includes stratification of the weight calculation. 
The topic of naive variance estimation using the standard sandwich variance estimator is considered for the three approaches. Some insights into the biases of the three approaches when the independent censoring assumption is violated are also gained.

In section~\ref{sec:regression_analysis}, the setting and main results on, in particular, the asymptotic variances are presented. In section~\ref{sec:theoretical_example}, a theoretical example offers insights into how the three approaches compare in a specific setting. In section~\ref{sec:simulation}, simulations are used to illustrate, corroborate, and challenge the asymptotic results. 
The paper ends with a discussion of results in section~\ref{sec:discussion} and an appendix with technical results in section~\ref{sec:appendix}.

\section{Regression analysis with a censored outcome}
\label{sec:regression_analysis}

Suppose $Y$ is an outcome of interest and a model is assumed which states that a parameter vector $\beta \in \R^p$ exists such that $\E(Y \given X) = \mu(\beta; X)$ for a certain function $\mu$ where $X$ is a vector of covariates. It is desired to estimate the true $\beta$ based on $n$ observations, i.e.\ $n$ independent replications of that experiment. 
A standard approach would be to solve an estimating equation of the type
\begin{equation} \label{eq:Un_original}
U_n(\beta) := \sum_{i=1}^n A(\beta; X_i) (Y_i - \mu(\beta; X_i)) = 0
\end{equation}
for a suitable $p$-dimensional vector function $A$. 
When $Y$ is not always observed due to censoring, this approach can no longer be taken. 
To be more specific, a competing risks setting is now considered. Suppose $Y$ is a function of a failure time $T > 0$ and a failure type $D \in \{1, \dots, d\}$, but determined by a time point $t>0$. In other words $Y = y(T \wedge t, D \indic{T \leq t})$ for a reasonable function $y$. Outcomes of this type include 
\begin{align}
Y &= \indic{T > t}, \quad & & \text{survival to time } t. \\
Y &= \indic{T \leq t, D = j}, \quad& & \text{failure of a specific type before time } t. \\
Y &= T \wedge t, \quad & & \text{survival time restricted to time } t. \\
Y &= (t - T \wedge t) \indic{D = j}, \quad & & \text{time lost to a specific failure before } t.
\end{align}
More generally, $T$ may be some event time and $D$ an event type. 
Possible models of $\E(Y\given X)$ include examples from generalized linear models: a linear model $\mu(\beta; x) = \beta\T x$; a relative or exponential model, $\mu(\beta; x) = \exp(\beta \T x)$; or a logistic model  $\mu(\beta; x) = 1/(1+\exp(-\beta\T x))$. 
The logistic model is primarily appropriate for the dichotomous outcomes. 
In all examples, the vector $x$ likely includes a constant term and should more generally be considered a result of a vector function applied to an original set of covariates since it should also be able to hold interaction terms, for instance.
In the generalized linear model setting, the function $A$ depends on the choice of link function, which determines the structure of $\mu$, and on the choice of family. The Gaussian family corresponds to the choice of $A(\beta;x) = \frac{\partial}{\partial \beta} \mu(\beta;x)$. A simple choice is $A(\beta;x)  = x$ and this is obtained when the link function used is canonical for the specified family.

Next, suppose $T$ and $D$ are not always observed due to censoring at a censoring time $C > 0$. Instead $\tilde T = T \wedge C$, the observed exit time, and $\tilde D = D \indic{T \leq C}$, the observed exit type with 0 denoting censoring, are observed. An outcome available at time $t$ as above is observed if $C \geq T \wedge t$ and can in this case be written as $Y = y(\tilde T \wedge t, \tilde D \indic{\tilde T \leq t})$.
To handle censoring, different weighting approaches are considered. The weights are allowed to depend on covariates only through a categorization $Z$ of the covariates $X$. 
If $G$ denotes the conditional survival function of the censoring time $C$, that is, $G(s \given z) = \P(C > s \given Z = z)$, a weight to consider is then
\begin{equation}
W = \frac{\indic{C \geq T \wedge t}}{G(T \wedge t- \given Z)} = \frac{\indic{\tilde T \geq t} + \indic{\tilde T < t, \tilde D \neq 0}}{G(\tilde T \wedge t- \given Z)}.
\end{equation}
The exact $G$ may well be unknown, and here an approach where an estimate of the censoring distribution is used instead is considered. 
In the following, estimation based on $n$ independent replications of this type of experiment where information on $(\tilde T, \tilde D, X)$ is available is considered. This setting, where censoring times are also censored by event times, allows for the use of the Kaplan--Meier estimator of $G$ within each stratum of $Z$. In order to handle ties systematically and appropriately according to the described setting where the event time $T$ takes priority over the censoring time $C$, a slight variation of the Kaplan--Meier estimator will in fact be used and is defined precisely in the appendix in equation~\eqref{eq:G_est}. This estimate is called $\hat G$. So, the applied weight will instead be
\begin{equation}
\hat W = \frac{\indic{C \geq T \wedge t}}{\hat G(T \wedge t- \given Z)} = \frac{\indic{\tilde T \geq t} + \indic{\tilde T < t, \tilde D \neq 0}}{\hat G(\tilde T \wedge t- \given Z)}.
\end{equation}

Three regression approaches based on the weights $\hat W$ are now considered: weighting the individual contribution and solving
\begin{equation}
U_{n,\textit{ind}}(\beta) := \sum_{i=1}^n A(\beta; X_i) \hat W_i (Y_i - \mu(\beta; X_i)) = 0,
\end{equation}
weighting only the potentially censored outcome and solving
\begin{equation}
U_{n,\textit{out}}(\beta) := \sum_{i=1}^n  A(\beta; X_i) (\hat W_i Y_i - \mu(\beta; X_i)) = 0,
\end{equation}
and replacing the outcome with a jack-knife pseudo-observation, $\hat \theta_i$, from a weighted estimator and solving 
\begin{equation}
U_{n,\textit{pse}}(\beta) := \sum_{i=1}^n  A(\beta; X_i) (\hat \theta_i - \mu(\beta; X_i)) = 0.
\end{equation}
The pseudo-observation will be defined as
\begin{equation}
\hat \theta_i = n \hat \theta - (n-1) \hat \theta^{(i)}
\end{equation}
where $\hat \theta$ is the overall estimate of $\E(Y)$, 
\begin{equation} \label{eq:theta_estimate}
\hat \theta = \frac{1}{n} \sum_{j=1}^n \hat W_j Y_j = \frac{1}{n} \sum_{j=1}^n \frac{\indic{C_j \geq T_j \wedge t}}{\hat G(\tilde T_j \wedge t- \given Z_j)} Y_j,
\end{equation}
and $\hat \theta^{(i)}$ is the estimate obtained by leaving out observation $i$, which can be written
\begin{equation}
\hat \theta^{(i)} = \frac{1}{n-1} \sum_{j\neq i} \hat W_j^{(i)} Y_j = \frac{1}{n-1} \sum_{j\neq i} \frac{\indic{C_j \geq T_j \wedge t}}{\hat G^{(i)}(\tilde T_j \wedge t- \given Z_j)} Y_j
\end{equation}
if $\hat G^{(i)}$ is the stratified Kaplan--Meier estimate of $G$ where observation $i$ has been left out.
Since $\hat G^{(i)}(s \given z) = \hat G(s \given z)$ when observation $i$ does not belong to stratum $z$, the jack-knife pseudo-observation of the overall estimator of \eqref{eq:theta_estimate} can also be calculated as the jack-knife pseudo-observation of the within-stratum estimate and equals 
\begin{equation}
    \hat \theta_i = \hat W_i Y_i + \sum_{j\neq i: Z_j = Z_i} (\hat W_j - \hat W_j^{(i)}) Y_j.
\end{equation}
In other words, the pseudo-observation includes the weighted outcome from before and an additional term that takes a potential influence of the observation on the estimate of the weight into account. 

It may be noted that the inverse probability weighted estimates of \eqref{eq:theta_estimate} include as examples the Kaplan--Meier estimate for the outcome $Y = \indic{T > t}$, the area under the Kaplan--Meier curve for the outcome $Y = T \wedge t$, and other common estimators of the mentioned examples of outcomes $Y$. \cite{satten2001kaplan} established this type of result for the Kaplan--Meier-based estimate of a failure probability.

Generally, the following assumptions on the censoring mechanism are made.
\begin{assumption} \label{ass:indep_cens}
Within strata of $Z$, the censoring time $C$ is independent of event time $T$, type $D$, and covariates $X$. Symbolically, $C \independent (T, D, X) \given Z$.
\end{assumption}

\begin{assumption} \label{ass:posi}
It is possible to observe the information of interest in all relevant strata of $Z$. That is, $G(t- \given z) > 0$ for (almost) all $z$. 
\end{assumption}

The assumptions ensure that $W$ is well defined and that $\E(W \given T, D, X) = 1$, which make the weights suitable for compensating for the missing information. 
The cumulative censoring hazard can also be defined without issue in the relevant interval as follows. Let $F_0(s \given z) = \P(C \leq s \given Z=z)$. Define the corresponding cumulative censoring hazard by $\Lambda(s \given z) = \int_0^s G(u- \given z)^{-1} F_0(\d u \given z)$ for $s \leq t$. No assumption of continuity of these functions is made. 

Other assumptions are made in the following. These assumptions include some usual requirements for the estimating procedure of the uncensored case in \eqref{eq:Un_original} to work, and a further positivity assumption, say $\P(T > t \given Z) > 0$ almost surely, to ease the handling of the estimation of the censoring distribution.
To be vague, these assumptions are collectively termed \textit{regularity conditions} in what follows.
Some further details on this matter are given in the appendix.
The approach presented in the appendix even puts restrictive assumptions on the outcome $Y$, in particular boundedness. These assumptions are met by the four presented examples of outcomes, but do not seem strictly necessary for the results presented in the following.

Under regularity conditions, the original estimating equation, \eqref{eq:Un_original}, based on uncensored information has, with a high probability for large $n$, as solutions consistent and asymptotically normally distributed estimates of the true $\beta$, which will be denoted $\beta_0$ in the following. This can be seen from results on $Z$-estimators, see for instance Chapter~5 of \cite{Vaart1998}. To be more specific, the estimates will be asymptotically linear with influence function 
\begin{equation}
    \dot \beta(T, D, X) =  -B(X) (Y - \mu(\beta_0; X)),
\end{equation}
where $B(X) = J(\beta_0)^{-1} A(\beta_0; X)$ and $J(\beta) = \E( - A(\beta; X) \frac{\partial}{\partial \beta} \mu(\beta; X))$.
This means that
\begin{equation}
    \sqrt{n}(\hat \beta_n - \beta_0) = \sqrt{n}\frac{1}{n} \sum_{i=1}^n \dot \beta(T_i, D_i, X_i) + o_{\P}(1),
\end{equation}
so that the asymptotic variance is $\Var(\dot \beta(T, D, X))$.

Under similar regularity conditions, and the assumptions on the censoring mechanism mentioned above, the three weighting approaches similarly have as solutions estimates that are asymptotically linear with a specific influence function. The influence functions compare to the influence function of the uncensored problem in a certain way, as laid out in the theorem below.
The notation $M(s \given Z) = \indic{C \leq s} - \int_0^s \indic{C \geq u} \Lambda(\d u \given Z)$ will be used for a martingale related to the censoring.

\begin{theorem} \label{thm:influence}
Under Assumption~\ref{ass:indep_cens}, Assumption~\ref{ass:posi}, and regularity conditions, the three approaches have as solutions consistent and asymptotically normal parameter estimates with influence functions on the form
\begin{equation} \label{eq:influence_function_beta_type}
\begin{aligned}
\dot \beta_\textup{type}(\tilde T, \tilde D, X) &= \dot \beta(T, D, X) \\
& \hspace{-6em} + \int_0^{t-} \big(\phi_\textup{type}(s; T, D, X) - \nu_\textup{type}(s \given Z) \big) \indic{T > s} \frac{1}{G(s \given Z)} M(\d s \given Z),
\end{aligned}
\end{equation}
where 
$\nu_\textup{type}(s \given Z) = \E(\phi_\textup{type}(s; T, D, X) \given T > s, Z)$. The three types have
\begin{align}
\phi_{ind}(s; T, D, X) &= B(X) (Y- \mu(\beta_0; X))    \\
\phi_{out}(s; T, D, X) &= B(X) Y \\
\phi_{pse}(s; T, D, X) &= B(X) (Y- \E(Y \given T > s, Z)).
\end{align}
\end{theorem}

The proof can be found in the appendix.

The last term of \eqref{eq:influence_function_beta_type} is structured with a part depending only on the underlying competing risks data and its distribution and a part depending only on the underlying censoring time and its distribution. 
Under the assumptions, this structure implies a similar structure of the resulting asymptotic variance matrices which makes clear how the variance has been increased by censoring and how the variance depends on the censoring distribution.

\begin{corollary} \label{cor:as_variance_express}
In the setting of Theorem~\ref{thm:influence}, the asymptotic variances $\Sigma_\textup{type} = \Var(\dot \beta_\textup{type}(\tilde T, \tilde D, X))$ can be expressed as
\begin{equation} \label{eq:Sigma_type}
\Sigma_\textup{type} = \Sigma + \E\big( \int_0^{t-} \Phi_\textup{type}(s \given Z) \frac{S(s \given Z)}{G(s \given Z)} \Lambda(\d s \given Z) \big)
\end{equation}
where $\Sigma = \Var(\dot \beta(T, D, X))$ is the variance of the uncensored problem, $\Phi_\textup{type}(s \given Z) = \Var(\phi_\textup{type}(s; T, D, X) \given T > s, Z)$, and $S(s \given z) = \P(T > s \given Z = z)$. Concretely, 
\begin{align}
\Phi_{ind}(s \given Z) &= \Var(B(X) (Y - \mu(\beta_0; X)) \given T > s, Z)   \\
\Phi_{out}(s \given Z) &= \Var(B(X) Y \given T > s, Z) \\
\Phi_{pse}(s \given Z) &= \Var(B(X) (Y - \E(Y \given T > s, Z)) \given T > s, Z).
\end{align}
\end{corollary}
\begin{proof}
The two terms of $\dot \beta_{\textup{type}}(\tilde T, \tilde D, X)$ in \eqref{eq:influence_function_beta_type} are uncorrelated owing to the independent censoring assumption: A martingale property applies to $M(s \given Z)$ given the underlying information $(T, D, X)$ and the last term will have mean 0 in the conditional distribution given $(T, D, X)$. The martingale property also implies that the variance of the last term is
\begin{equation}
\begin{aligned}
&\E\big(\int_0^{t-} \big(\phi_\textup{type}(s; T, D, X) - \nu_\textup{type}(s \given Z) \big)^{\otimes 2} \\
& \hspace{4em} \cdot \frac{\indic{T > s} \indic{C \geq s} (1- \Delta \Lambda(s \given Z))}{G(s \given Z)^2} \Lambda(\d s \given Z) \big)
\end{aligned}
\end{equation}
which reduces to the desired expression under the independent censoring assumption since it is the case that $\E((\phi_\textup{type}(s; T, D, X) - \nu_\textup{type}(s \given Z) )^{\otimes 2} \indic{T > s} \given Z) = \Var(\phi_\textup{type}(s; T, D, X) \given Z) S(s \given Z)$ and $\E(\indic{C \geq s} \given Z) (1- \Delta \Lambda(s \given Z)) = G(s- \given Z) (1- \Delta \Lambda(s \given Z)) = G(s \given Z)$.
See for instance Chapter II of \cite{Andersen1993} for implications of the martingale property.
Above, the notation $a^{\otimes 2} = a a \T$ is used for a column vector $a$.
\end{proof}

With the asymptotic variances at hand, it is of interest to consider the question of which of the types has the lower asymptotic variance and can therefore be expected to produce the lowest variance in parameter estimates, at least at larger sample sizes.
The answer clearly depends on many things. In particular it depends on how the three $\Phi_\textup{type}(s \given Z)$ compare at various time points $s$, but also on the censoring hazard at these time points. 
\begin{observation}
Some observations concerning the comparison of the asymptotic variances of the three approaches are the following.
\begin{enumerate}[label=(\alph*)]
    \item By the law of total variance, a lower bound of the asymptotic variances of the three approaches is formed by replacing $\Phi_\textup{type}(s \given Z)$ by $\Phi(s \given Z) = \Var(B(X)(Y- \E(Y \given T > s, X)) \given T > s, Z)$.
    
    \item For the individual weighting approach, $\mu(\beta_0;X)$ should be close to $\E(Y \given T > s, X)$ for $s$ close to 0 if the model holds. As a consequence, $\Phi_\textit{ind}(s \given Z)$ is close to the lower bound $\Phi(s \given Z)$ for $s$ close to 0. If the censoring hazard is high early and low later, the individual weighting approach should produce a comparably low variance.

    \item The outcome weighting approach should similarly have $\Phi_\textit{out}(s \given Z)$ close to the lower bound $\Phi(s \given Z)$ if $\E(Y \given T > s, X)$ is close to 0. This may happen for instance for the outcome examples of failure and time lost before $t$ if $s$ is close to $t$. If the censoring hazard is low early on, but high when approaching the time point of interest, $t$, the outcome weighting approach should produce a comparably low variance for this type of outcome, at least under continuity.

    \item Similarly, the pseudo-observation approach should have $\Phi_\textit{pse}(s \given Z)$ close to the lower bound $\Phi(s \given Z)$ if $\E(Y \given T > s, Z)$ is close to $\E(Y \given T > s, X)$. This can be expected to happen when the stratification, $Z$, is fine or when the outcome does not depend much on the covariates $X$. The examples given for the outcome weighting approach equally apply for the pseudo-observation approach. In fact, under continuity, $\E(Y \given T > s, Z)$ is expected to be close to $\E(Y \given T > s, X)$ for $s$ approaching $t$ no matter the outcome type since $Y$ is determined by time $t$; they are for instance both close to 1 for the survival outcome. 

    \item In the case of categorical covariates $X$, and $Z=X$, all three approaches have the same asymptotic variance. In fact, it seems often even the parameter estimates will all be the same according to a result of section~2.5 of \cite{blanche2022logistic} and a similar property for the pseudo-observations in line with results by \cite{stute1994jackknife}.
    In this light, we may expect similar asymptotic variances when more general covariates are considered and a fine stratification is used.
\end{enumerate}
\end{observation}

Overall, the pseudo-observation approach may seem to have the best chance of producing a low variance, but these are of course rather vague observations.
An example corroborating this observation is given in section~\ref{sec:theoretical_example}.

The outcome of failure is considered in detail by \cite{blanche2022logistic} and their Corollary~2 states that the difference in asymptotic variance matrix between the individual weighting approach and the outcome weighting approach can be negative definite or positive definite depending on the censoring distribution, which is in line with the observations above. It seems that this conclusion cannot immediately be transferred to, to give an example, the outcome of restricted survival time since $\E(Y \given T > s, X)$ for $s$ approaching $t$ would be close to $t$ rather than 0 in this case. As alluded to above, a similar conclusion can be reached in the comparison of the individual weighting approach and the pseudo-observation approach even for an outcome such as the restricted survival time.

A standard approach to estimating the variance of parameter estimates is to use the corresponding Huber--White-type sandwich variance estimator.
For each type, $U_{n,\textup{type}}(\beta)$ is on the form $\sum_{i=1}^n u_{i,n,\textup{type}}(\beta)$. The corresponding standard sandwich estimate of the asymptotic variance is
\begin{equation}
    n(\frac{\partial}{\partial \beta} U_{n,\textup{type}}(\beta))^{-1} \sum_{i=1}^n u_{i,n,\textup{type}}(\beta) u_{i,n,\textup{type}}(\beta)\T (\frac{\partial}{\partial \beta} U_{n,\textup{type}}(\beta)\T)^{-1}
\end{equation}
evaluated at the corresponding $\beta$ estimate.
The next result establishes that the standard sandwich variance estimate will be conservative for large $n$ for all three approaches when the model holds.
\begin{theorem} \label{thm:varest}
In the setting of Theorem~\ref{thm:influence}, for each of the three approaches, the standard sandwich variance estimator converges in probability to
\begin{equation}
\Sigma_\textup{type}' = \Sigma + \E\big( \int_0^{t-} \Phi_\textup{type}'(s \given Z) \frac{S(s \given Z)}{G(s \given Z)} \Lambda(\d s \given Z) \big)
\end{equation}
where $\Phi_\textup{type}'(s \given Z) = \E(\phi_\textup{type}(s; T, D, X) \phi_\textup{type}(s; T, D, X) \T \given T > s, Z)$, and consequently, for any of the three types, $\Sigma_\textup{type}' \geq \Sigma_\textup{type}$ with equality if and only if $\E(\phi_\textup{type}(s; T, D, X) \given T > s, Z = z) = 0$ for $\Lambda(\cdot \given z)$-almost all $s$ for almost all $z$.
\end{theorem}

The proof of the theorem can be found in the appendix.

This result is in line with and extends Proposition 1 of \cite{blanche2022logistic}.

\begin{observation}
A few observations concerning the asymptotic variance and variance estimation are given in the following.
\begin{enumerate}[label=(\alph*)]
    \item The requirement for equality does not seem particularly reasonable in applications in any of the three cases, except perhaps for $\phi_\textit{pse}$ in cases where  $\E(Y  \given T > s, Z)$ is close to $\E(Y \given T > s, X)$.
    In the simple case where $X$ itself represents strata and can be obtained from $Z$, for instance $X=Z$, equality is generally only obtained for the pseudo-observation approach. 
    \item At least for the individual and outcome weighting approaches, the limit $\Sigma_{\textup{type}}'$ corresponds to what $\Sigma_{\textup{type}}$ would have been if $G$ had been postulated as the censoring distribution rather than estimated  using the Kaplan--Meier estimator. 
    \item It should be perfectly possible to estimate the asymptotic variance by estimating the influence function, plugging in the observed data points and evaluating the empirical variance of what is obtained. The expression of the influence functions in Theorem~\ref{thm:influence} is less helpful here, but the results and approaches of the appendix may be. Such an alternative variance estimator may however have its own problems in small samples where estimating the asymptotic variance is less relevant. 
    This was seen in \cite{Overgaard2018} in an example of the pseudo-observation approach.
    \item Suppose for a moment that the censoring distribution is the same in the various strata, $C \independent Z$, in addition to the independent censoring assumption, Assumption~\ref{ass:indep_cens}. For the individual and outcome weighting approaches, it can be shown that a finer stratification reduces the asymptotic variance $\Sigma_{\textup{type}}$. The limit of the sandwich variance estimator,  $\Sigma_{\textup{type}}'$, is however left unchanged, and in this way the sandwich variance estimator should not be expected to pick up on the actual advantage.
    The situation is more unclear for the pseudo-observation approach, but it seems a finer stratification will tend to reduce the asymptotic variance $\Sigma_{\textup{type}}$ as well as the limit $\Sigma_{\textup{type}}'$ if chosen appropriately to try to match $\E(Y \given T > s, Z)$ with $\E(Y \given T > s, X)$.
\end{enumerate}
\end{observation}

The stated results assume that the regression model $E(Y \given X) = \mu(\beta; X)$ holds for some $\beta$, also denoted $\beta_0$. 
As can be seen from the results of the appendix, Lemma~\ref{lemma:estimator_exist} and Lemma~\ref{lemma:rep_independent} specifically, the results do not change much under misspecification of the regression model. 
The three approaches will be able to estimate a best fit to the uncensored problem which may be useful in some situations. The best fit will then depend on the covariate distribution.  
Under misspecification of the regression model, the only real differences to the results stated earlier are that $\beta_0$ should now refer to the best fit rather than a true $\beta$ and that $J(\beta_0)$ and thereby $B(X)$ will have a more complicated expression, which can be found in equation~\eqref{eq:J_general} of the appendix.
This does not change the observations made above much except perhaps for one observation concerning the individual weighting approach: The $\E(Y \given T > s, X) $ can no longer be expected to be as close to $\mu(\beta_0; X)$ as before for $s$ close to $0$ and so $\Phi_\textit{ind}(s \given Z)$ can no longer be expected to be as close to the lower bound $\Phi(s \given Z)$ as before.
Concerning variance estimation, statistical program packages may or may not use an appropriate estimate of the more general and complicated expression of $J(\beta_0)$ relevant under misspecification of the regression model when producing the standard sandwich variance estimate. If they do, Theorem~2 will apply with the mentioned modifications. If they do not, there will be another source of bias of the variance estimate, which will not be examined more closely here. 
The discussion concerning the more complicated and the more simple expression of $J(\beta_0)$ relates to the discussion of whether to use the observed information matrix or an estimate of the expected information matrix under the model. 

A source of bias in the regression parameter estimates is violation of the independent censoring assumption. 
An attempt at examining the bias is made in Proposition~\ref{prop:bias} of the appendix where a first order approximation of the bias is found for the three approaches. In principle, this gives some insights into how the bias compares between the three approaches, but the result is approximate and the conclusion is not particularly clear. 
One observation concerns the scenario where $(T,D) \independent C \given X$ but not $(T,D,X) \independent C \given Z$ holds. Here, equation~\eqref{eq:bias} of the appendix reveals how discrepancy between the true and fitted censoring hazard, according to the approximation,  
\begin{enumerate}[label=(\alph*)]
    \item will not contribute to bias at early time points close to $0$ for the individual weighting approach if the regression model holds,
    \item will not contribute to bias at late time points close to $t$ for the outcome weighting approach for the outcome of failure and time lost before $t$,
    \item will not contribute to bias at late time points close to $t$ for the pseudo-observation approach generally.
\end{enumerate}
Overall, it seems the pseudo-observation approach will have the best chance of mitigating bias from this sort of violation of the independent censoring assumption by choosing $Z$ to be predictive of the outcome.

\section{A theoretical example}
\label{sec:theoretical_example}

With an aim of comparing the asymptotic variances precisely, an example in a simple setting is now considered.

Suppose covariates consists of two groups expected to be equal in size, $\P(X=0) = \P(X=1) = 1/2$. The event time is chosen to follow uniform distributions in the two groups,
\begin{align*}
\P(T \leq s \given X = 1) &= ps, \quad s \in [0, \frac{1}{p}], \\
\P(T \leq s \given X = 0) &= qs, \quad s \in [0, \frac{1}{q}],
\end{align*}
for certain choices of $p, q \in (0,1)$. 
The outcome of interest is $Y = \indic{T \leq 1}$, corresponding to a time point of interest $t=1$. A true model is $\mu(\beta; X) = \beta_0 + \beta_1 X$ with $\beta_0 = q$ and $\beta_1 = p-q$.
Censoring occurs with probability $0.5$ at a specific time point $s$ only and this happens independently of other variables. 

The three approaches are used for estimation of the $\beta$ parameters with the choice $A(\beta; X)=(1, X)\T$. 
No stratification will be used.
The asymptotic variances for the three types are
\begin{equation}
    \Sigma_\textup{type} = \Sigma + \Phi_\textup{type}(s) S(s)
\end{equation}
in this case according to Corollary~\ref{cor:as_variance_express}. Differences in asymptotic variances, particularly the signs of differences, are in this way given from differences in $\Phi_\textup{type}(s)$.

The matrix $J(\beta)$ is 
\begin{equation}
    J(\beta) = \E \Big( \begin{pmatrix} 1 & X \\ X & X \end{pmatrix} \Big) = \begin{pmatrix} 1 & \frac{1}{2} \\ \frac{1}{2} & \frac{1}{2} \end{pmatrix}
\end{equation}
such that
\begin{equation}
    J(\beta)^{-1} = \begin{pmatrix} 2 & -2 \\ -2 & 4 \end{pmatrix}.
\end{equation}
Calculations reveal
\begin{align*}
    f_1(s):= \E(Y \given T > s) &= \frac{(p+q)(1-s)}{2 - ps - qs}, \\
    f_2(s):= \Cov(Y, X \given T > s) &= \big(\frac{p(1-s)}{1-ps} - \frac{(p+q)(1-s)}{2 - ps - qs} \big) \frac{1-ps}{2 - ps - qs}, \\
    f_3(s) := \Cov(YX, X \given T > s) &= \frac{p(1-s)}{2 - ps - qs} \frac{1-qs}{2 - ps - qs}, \\
    f_4(s) := \Var(X \given T > s) &= \frac{1-ps}{2 - ps - qs} \frac{1-qs}{2 - ps - qs}.
\end{align*}
The asymptotic variance of the $\beta_1$ component is considered by applying the vector $a = (0, 1)\T$. Calculations reveal, for a general $a$,
\begin{align*}
    a\T \Phi_{\textit{ind}}(s) a &= a\T \Phi_{\textit{out}}(s) a + \Var(a\T B(X) \mu(\beta; X) \given T > s) \\
    & \quad - 2 \Cov(a \T B(X) Y, a \T B(X) \mu(\beta; X) \given T > s), \\
    a\T \Phi_{\textit{pse}}(s) a &= a\T \Phi_{\textit{out}}(s) a + f_1(s)^2 \Var(a \T B(X) \given T > s) \\
    & \quad - 2 \Cov(a \T B(X) Y, a \T B(X)) f_1(s).
\end{align*}
For the particular choice of vector $a = (0, 1)\T$, 
\begin{align*}
    \Var(a\T B(X) \mu(\beta; X) \given T > s) &= 4 (p+q)^2 f_4(s), \\
    \Cov(a \T B(X) Y, a \T B(X) \mu(\beta; X) \given T > s) &= 8 (p+q) f_3(s) - 4 (p+q) f_2(s), \\
    \Var(a \T B(X) \given T > s) &= 16 f_4(s) \\
    \Cov(a \T B(X) Y, a \T B(X) \given T > s) &= 16 f_3(s) - 8 f_2(s)
\end{align*}
and thereby
\begin{align*}
    &a\T (\Phi_{\textit{pse}}(s) - \Phi_{\textit{out}}(s) ) a = 16 f_1(s)(f_1(s) f_4(s) - 2 f_3(s) + f_2(s)) \\
    &= - 16 \frac{(p+q)(1-s)}{2 - ps - qs} \big( \frac{q(1-s)}{2 - ps - qs} \big(\frac{1-ps}{2 - ps - qs}\big)^2 + \frac{p(1-s)}{2 - ps - qs} \big(\frac{1-qs}{2 - ps - qs}\big)^2\big)
\end{align*}
which is negative for all $s \in (0,1)$, indicating an asymptotic advantage of the pseudo-observation approach over the outcome weighting approach in this setting.
Also,
\begin{align*}
    a\T (\Phi_{\textit{ind}}(s) - \Phi_{\textit{out}}(s) ) a &= 4 (p+q)^2 f_4(s) - 16 (p+q) f_3(s) + 8 (p+q)f_2(s)) \\
    &= 4 \frac{p+q}{(2 - ps - qs)^2} \big( q(s(1-q)-(1-s))(1-ps) \\
    & \quad + p(s(1-p)-(1-s))(1-qs) \big).
\end{align*}
Certainly for $s < \min(\frac{1}{2-p},\frac{1}{2-q})$ the difference is negative, in favor of the individual weighting approach over the outcome weighting approach, while for $s > \max(\frac{1}{2-p},\frac{1}{2-q})$ the difference is positive, in favor of the outcome weighting approach.

The comparison of the pseudo-observation approach and the individual weighting approach may be obtained by subtracting the two expressions from each other. The difference is in favor of the pseudo-observation approach for $s=1$ and is 0 for $s=0$. Seemingly, the difference for any $s \in (0,1)$ will be an intermediate value and so in favor of the pseudo-observation approach over the individual weighting approach.

For the estimation of the intercept parameter, $\beta_0$, the vector $a= (1, 0) \T$ is instead considered. With this choice, 
\begin{align*}
    \Var(a\T B(X) \mu(\beta; X) \given T > s) &= 4 q^2 f_4(s), \\
    \Cov(a \T B(X) Y, a \T B(X) \mu(\beta; X) \given T > s) &= 4 q f_3(s) - 4 q f_2(s), \\
    \Var(a \T B(X) \given T > s) &= 4 f_4(s), \\
    \Cov(a \T B(X) Y, a \T B(X) \given T > s) &= 4 f_3(s) - 4 f_2(s),
\end{align*}
which leads to 
\begin{align*}
    &a\T (\Phi_{\textit{pse}}(s) - \Phi_{\textit{out}}(s) ) a = 4 f_1(s)(f_1(s) f_4(s) - 2 f_3(s) + 2 f_2(s)) \\
    &= 4 \frac{(p+q)(1-s)^2(1-ps)}{(2 - ps - qs)^3} \big( (p+q)\frac{1-qs}{2-ps-qs} - 2 q\big).
\end{align*}
This is in favor of the pseudo-observation approach over the outcome weighting approach when and only when $(p+q)(1-qs) < 2q (2-ps-qs)$ which happens when $s < \frac{3q - p}{q(p+q)}$. That can happen potentially never and potentially always in the interval $(0,1)$ depending on the values of $p$ and $q$.
Next, 
\begin{align*}
    a\T (\Phi_{\textit{ind}}(s) - \Phi_{\textit{out}}(s) ) a &= 4 q ( q f_4(s) - 2 f_3(s) + 2 f_2(s)) \\
    &= 4 q^2 \frac{1-ps}{2-ps-qs} \frac{s(2-q) - 1}{2-ps-qs}
\end{align*}
which is in favor of the individual weighting approach when $s(2-q)-1 < 0$, that is, when $s < \frac{1}{2-q}$.
With this result, examples can be found where the asymptotic variance of the $\beta_0$ estimate is smaller for the individual weighting approach than for the pseudo-observations approach. Take for instance $p=1/2$, $q=1/6$ and $s \in (0, 1/(2-\frac{1}{6}) )$.

This simple theoretical example illustrates how examples can be found where either of the three approaches will have the smallest asymptotic variance.

\section{Simulations}
\label{sec:simulation}

To gain insights into the behavior of the three approaches in finite samples, a simulation study has been conducted as described in the following.
Three separate scenarios are considered in this simulation study, each with different configurations. Each configuration is simulated 10 000 times.  

\paragraph{Scenario I -- cumulative incidence.}
This scenario considers the comparison of the cumulative incidence, or risk, in two groups by a risk difference.
The purpose is a comparison of three approaches in a simple setting at different censoring distributions and various sample sizes.
The simple setting corresponds to the theoretical example given in section~\ref{sec:theoretical_example}.
First, the group $X \in \{0,1\}$ is drawn with equal probability $\P(X=x) = 0.5$ in the two groups. The time to event, $T$, is drawn according to $\P(T \leq s \given X=x) = p_x s$ for $s \in (0, 1/p_x)$ where the choices $p_0 = 1/6$ and $p_1 = 1/2$ are used.
The time point $t=1$ is the time point of interest in the following, and the outcome $Y = \indic{T \leq 1}$ is the outcome of interest. 
A simple linear model, $\mu(\beta; X) = \beta_0 + \beta_1 X$ is considered such that $\beta_1$ is the risk difference and the unknown parameter to be estimated. 
The specification above makes $\beta_1 = p_1 - p_0 = -1/3$ the true value of the risk difference at the time point of interest. 
The censoring distribution will be independent of $T$ and $X$. Three censoring distributions are considered: one where about \SI{50}{\percent} are censored at the early time point 0.2, $P(C = 0.2) = 0.5$, and the rest remain uncensored; one where about \SI{50}{\percent} are censored at the late time point 0.8, $P(C = 0.8) = 0.5$, and the rest remain uncensored; and an exponential censoring distribution, $\P(C > s) = \exp(-s)$.
Samples of sizes $n \in \{50, 100, 200, 400, 800\}$ are considered. 
The three approaches are used unstratified, that is, the overall Kaplan--Meier estimator of the censoring distribution is used for the weights, and with $A(\beta; X) = (1, X)\T$. 
Note that the choice $\P(C = s) = 0.5$ and the rest remaining uncensored, or $\P(C \geq 1) = 0.5$, will have $\Delta \Lambda(s)/G(s) = 0.5/0.5 = 1$ and is able to pick out the remaining integrand in the last part of \eqref{eq:Sigma_type} at the chosen time point $s$, for instance $s=0.2$ or $s=0.8$. In other words, $\Sigma_\textup{type} = \Sigma + \Phi_\textup{type}(s) S(s)$ in this case.

\begin{figure}
    \centering
    \includegraphics[width=0.98\linewidth]{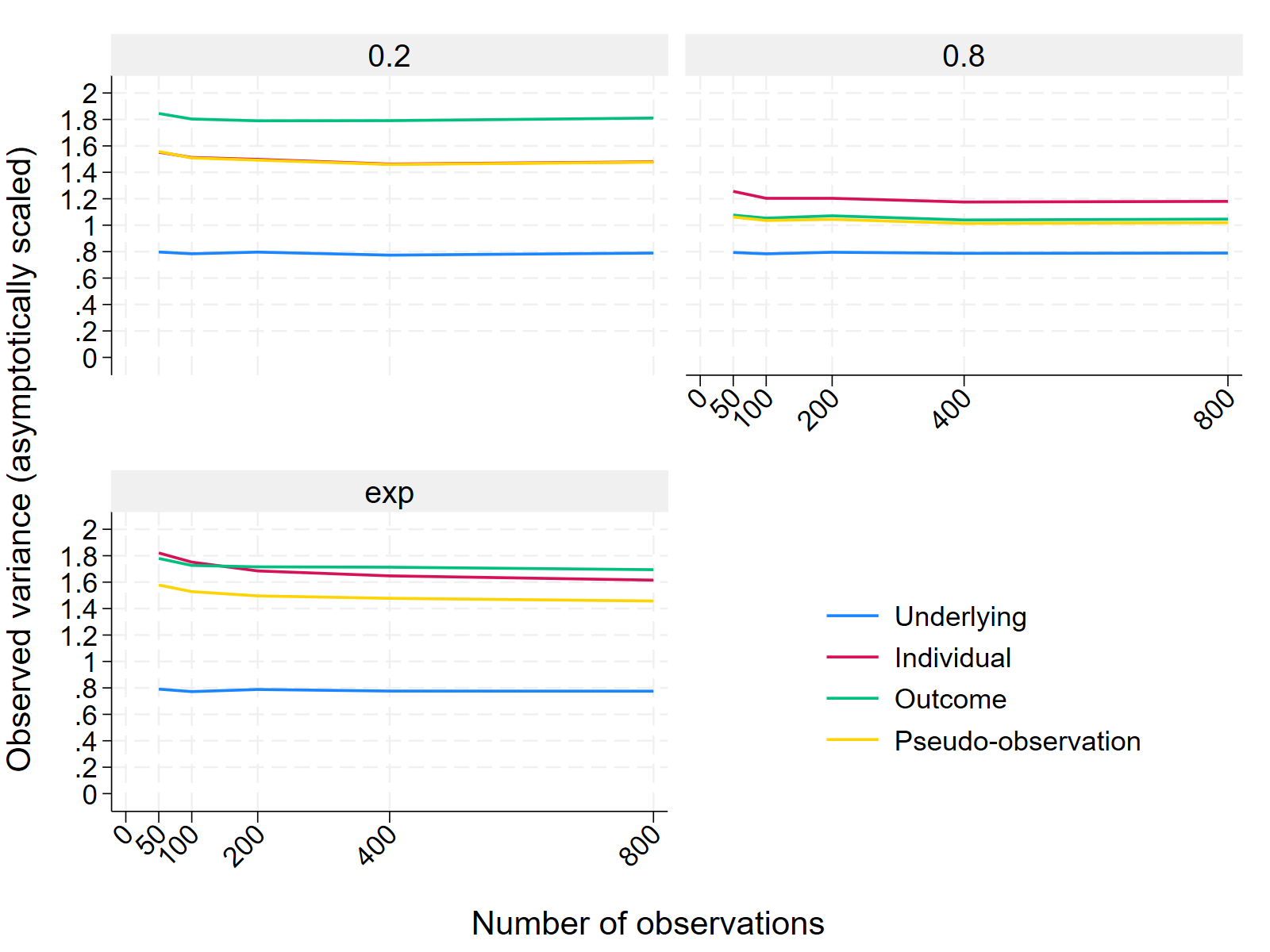}
    \caption{Observed variances of $\beta_1$ estimates scaled by $n$ for the three approaches, as well as for the underlying approach based on uncensored data, in each of the three censoring distribution settings in scenario~I.}
    \label{fig:sc1_var}
\end{figure}

\begin{table}[htbp]
\footnotesize
\centering
\begin{tabular}{rrrrrrrrrrr}
  \toprule
Cens & $n$ & $\Var_{ind}$ & $\Var_{out}$ & $\Var_{pse}$ & $\widehat{\Var}_{ind}$ & $\widehat{\Var}_{out}$ & $\widehat{\Var}_{pse}$ & $\%_{ind}$ & $\%_{out}$ & $\%_{pse}$ \\ 
  \midrule
  0.2 &  50 & 1.55 & 1.85 & 1.56 & 1.47 & 1.89 & 1.62 & 92.6 & 94.5 & 94.7 \\ 
  0.2 & 100 & 1.51 & 1.80 & 1.51 & 1.47 & 1.87 & 1.58 & 93.8 & 94.9 & 95.1 \\ 
  0.2 & 200 & 1.50 & 1.79 & 1.49 & 1.47 & 1.86 & 1.56 & 94.1 & 95.1 & 95.2 \\ 
  0.2 & 400 & 1.46 & 1.79 & 1.46 & 1.46 & 1.85 & 1.55 & 94.7 & 95.2 & 95.5 \\ 
  0.2 & 800 & 1.48 & 1.81 & 1.48 & 1.46 & 1.85 & 1.54 & 94.9 & 95.4 & 95.7 \\ 
  0.8 &  50 & 1.26 & 1.08 & 1.06 & 1.21 & 1.06 & 1.05 & 93.7 & 94.1 & 94.1 \\ 
  0.8 & 100 & 1.20 & 1.05 & 1.04 & 1.19 & 1.05 & 1.03 & 94.7 & 94.5 & 94.6 \\ 
  0.8 & 200 & 1.20 & 1.07 & 1.04 & 1.18 & 1.05 & 1.03 & 94.7 & 94.8 & 94.7 \\ 
  0.8 & 400 & 1.18 & 1.04 & 1.01 & 1.18 & 1.05 & 1.02 & 95.0 & 95.1 & 95.1 \\ 
  0.8 & 800 & 1.18 & 1.05 & 1.02 & 1.17 & 1.04 & 1.02 & 94.8 & 95.0 & 94.8 \\ 
  exp &  50 & 1.82 & 1.78 & 1.58 & 1.71 & 1.78 & 1.61 & 92.7 & 94.2 & 94.4 \\ 
  exp & 100 & 1.75 & 1.73 & 1.53 & 1.67 & 1.76 & 1.56 & 93.9 & 94.9 & 94.8 \\ 
  exp & 200 & 1.69 & 1.72 & 1.50 & 1.65 & 1.75 & 1.54 & 94.3 & 95.0 & 95.1 \\ 
  exp & 400 & 1.65 & 1.71 & 1.48 & 1.65 & 1.74 & 1.53 & 94.9 & 95.2 & 95.4 \\ 
  exp & 800 & 1.62 & 1.69 & 1.46 & 1.64 & 1.74 & 1.52 & 95.0 & 95.4 & 95.3 \\
   \bottomrule
\end{tabular}
\caption{Simulation results of Scenario I: Observed variance in $\beta_1$ estimates ($\Var$), mean of corresponding standard sandwich variance estimates ($\widehat \Var$), and coverage probability of Wald-type \SI{95}{\percent} confidence intervals based on standard sandwich variance estimates ($\%$) for each type of approach (\textit{ind}, \textit{out}, and \textit{pse}) and in each configuration of censoring distribution (Cens) and number of observations ($n$). Observed and estimated variances are scaled by $n$ (asymptotically scaled) for easy comparison.}
\label{tab:sc1}
\end{table}

Key results of the simulations in this scenario are summarized in Figure~\ref{fig:sc1_var} and Table~\ref{tab:sc1}, where the focus is on the variance of $\beta_1$ estimates.
Observed biases for $\beta_1$ in the approaches seem negligible and are not presented here. 
As is illustrated in Figure~\ref{fig:sc1_var}, the observed variance of the parameter estimates is in line with what is suggested by the theory and the theoretical example of Section~\ref{sec:theoretical_example}, especially for the larger sample sizes: the pseudo-observation approach produces a comparably low variance in this setting in all three censoring distribution configurations; the individual weighting approach produces a comparably low variance when censoring occurs at an early time point; the outcome weighting approach produces a reasonably low variance when censoring occurs at a late time point. 
It can also be seen how the individual weighting approach eventually produces a lower variance than the outcome weighting approach in the exponential censoring distribution configuration in this setting, which is not immediately clear from the theory already presented. 
The fact that losing information to censoring results in larger variances of parameter estimates, as seen in Corollary~\ref{cor:as_variance_express}, is illustrated by the gap from the observed variances of the three approaches to the observed variance of the underlying approach based on uncensored data. It seems reasonable that the gap is smaller when censoring occurs at a late time point since less information is lost.

Table~\ref{tab:sc1} gives further insights into the variance estimation using the standard, Huber--White-type sandwich variance estimator: In most cases the average variance estimate is larger than the observed variance in parameter estimates, as suggested by Theorem~\ref{thm:varest}. Somewhat surprisingly, the opposite does happen even for larger samples sizes, at least for the outcome weighting approach. A calculation reveals that the asymptotic difference is quite small in this case and it can apparently not be expected to show up in simulations at this sample size and number of iterations. Overall, the differences between observed variance and average variance estimate are fairly small, and the corresponding coverage probabilities of Wald-type \SI{95}{\percent} confidence intervals using these variance estimates are quite close to \SI{95}{\percent}, at least for the larger sample sizes.

To sum up this scenario, the pseudo-observation approach generally wins out in terms of producing low variance of the estimates of the important regression parameter in this case, and the standard sandwich variance estimator even produces reasonable variance estimates for the pseudo-observation approach, leading to reasonable coverage probabilities.

\paragraph{Scenario II - restricted mean, a misspecified model, and covariate-dependent censoring.}
In this scenario, the restricted mean for different values of continuous covariates are compared using differences. 
The purpose is a comparison of the three approaches with a new outcome type in a more complicated situation where there are more covariates, continuous covariates, misspecification of the regression model, and covariate-dependent censoring. 
Three independent continuous covariates are considered, $X=(X_1, X_2, X_3)$ where $X_1 \sim N(0,1), X_2 \sim U(0,1), X_3 \sim \Gamma(\text{shape } 3, \text{scale } 0.5)$.
Given $X$, the event time $T$ will follow a Weibull distribution with shape parameter 1.5 and rate parameter $\exp(-2 + X_1 + X_2/6 + X_3/2 + X_2 \cdot X_3/4)$.
The time point of interest is $t=1$ and the outcome of interest is then $Y = T \wedge 1$. 
The model considered is $\mu(\beta; X) = \beta_0 + \beta_1 X_1 + \beta_2 X_2 + \beta_3 X_3$. No attempt will be made to find $\E(Y \given X)$, but it is apparent that $\mu(\beta; X)$ is misspecified. 
Given $X$, the censoring distribution will be independent of $T$ and follow a Weibull distribution with shape parameter 1.5 and rate parameter $\exp(-0.5 + X_2)$. In particular, the censoring distribution depends on the uniformly distributed $X_2$.
The three approaches are used with $A(\beta; X) = (1, X_1, X_2, X_3)\T$ and the sample size $n=1000$ is considered throughout. 
It is expected that handling the covariate-dependent censoring by stratification will be useful. The stratification variable $Z$ is constructed from $X_2$ by first choosing a number of strata $k$, then letting $Z = j$ when $j/k < X_2 \leq (j+1)/k$ for $j = 0, \dots, k-1$. The three approaches are considered with $k \in \{1, 2, 4, 8\}$.
Note that although the model is misspecified, the fits from the three approaches should on average match the best fit from uncensored problem if the censoring mechanism is handled appropriately according to the theory. 

\begin{figure}
    \centering
    \includegraphics[width=0.98\linewidth]{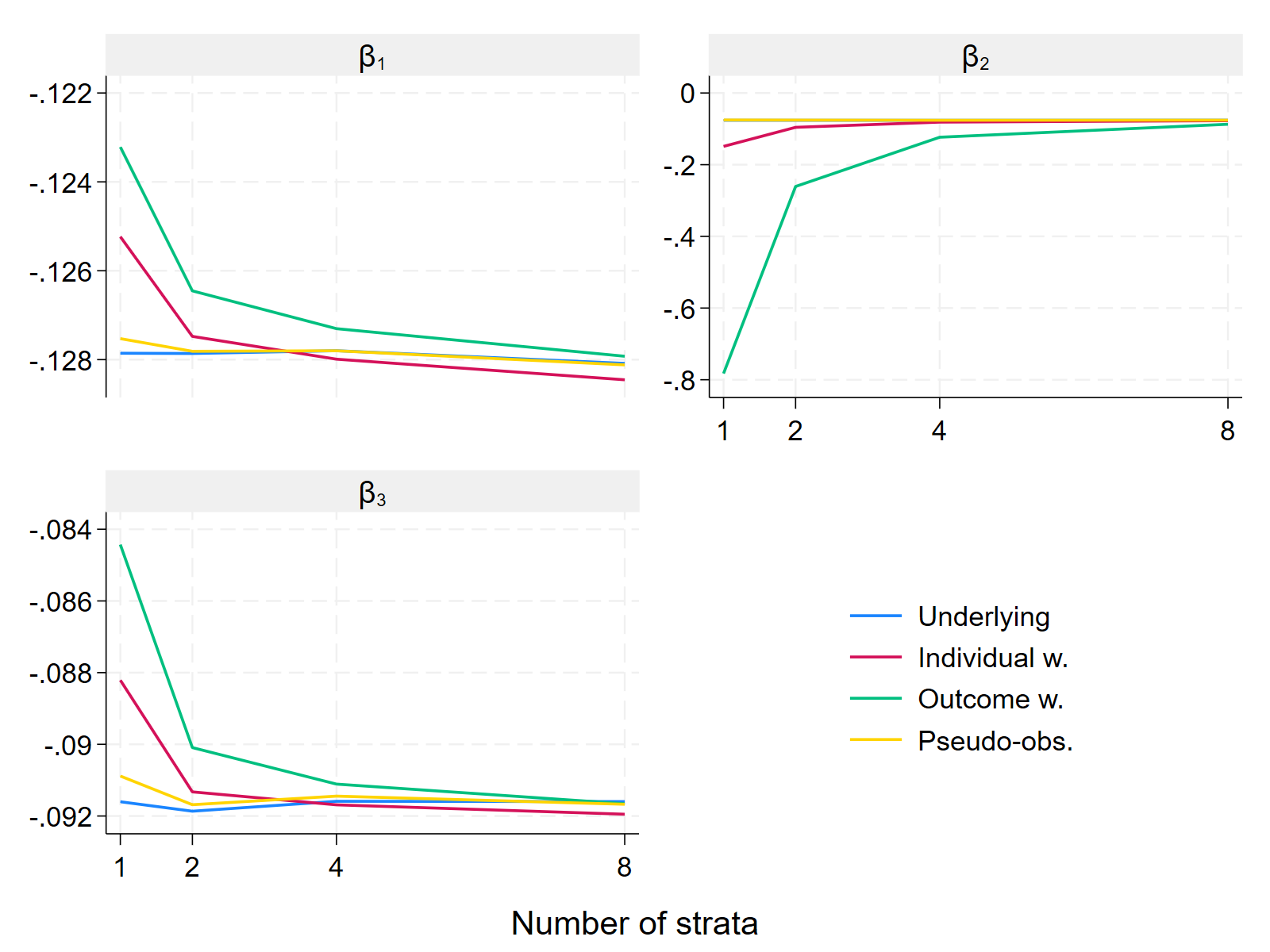}
    \caption{Average parameter estimates for each of three parameters, $\beta_1$, $\beta_2$, and $\beta_3$, for each of the three approaches as well as for the corresponding approach on the underlying uncensored data according to the number of strata used in the estimation of the censoring distribution in scenario~II.}
    \label{fig:sc2_b}
\end{figure}

\begin{figure}
    \centering
    \includegraphics[width=0.98\linewidth]{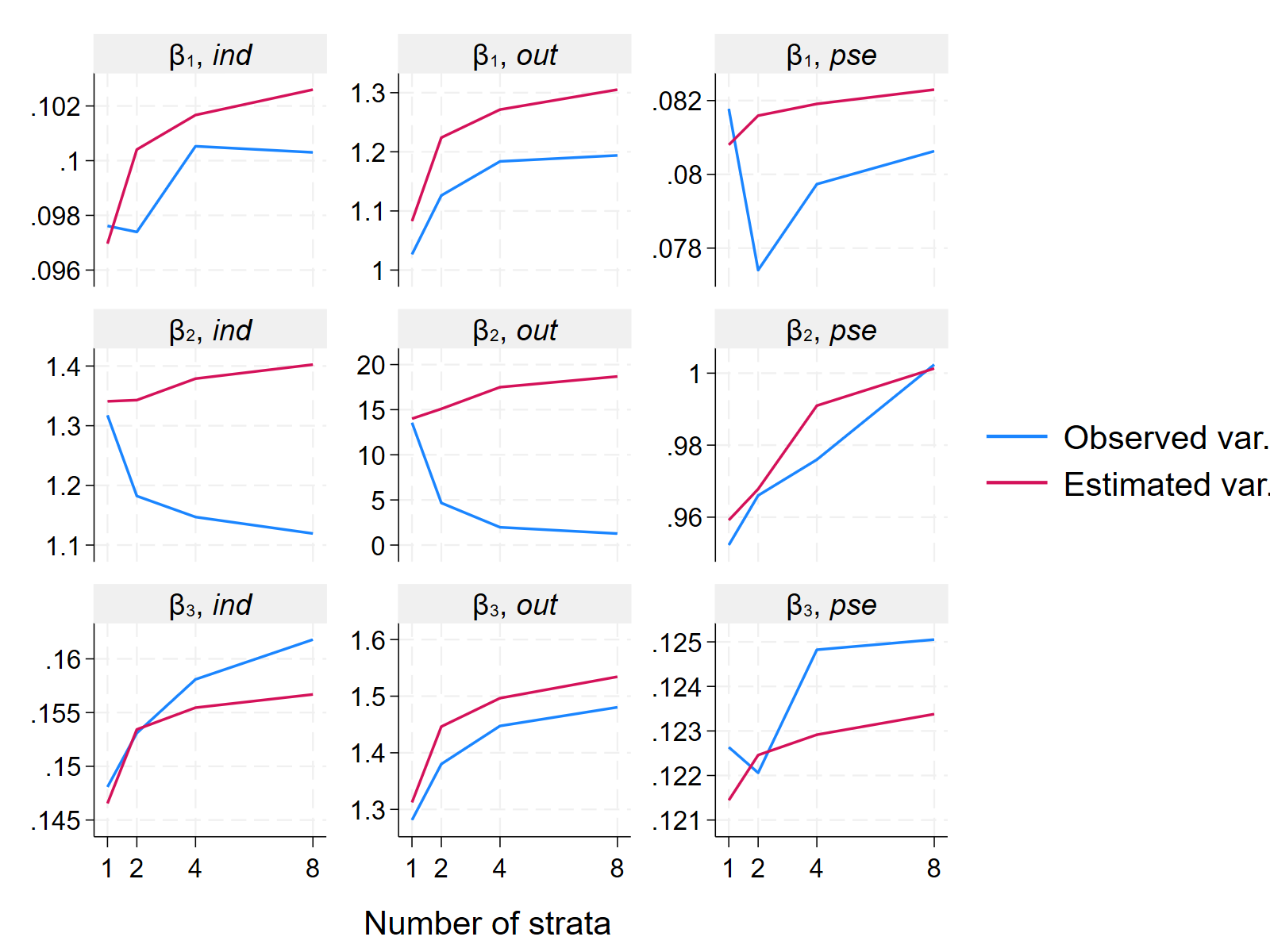}
    \caption{Observed variances and variances estimated by the standard sandwich variance estimator, both scaled by $n=1000$, for estimates of each of the three parameters $\beta_1$, $\beta_2$, and $\beta_3$ and for each of the three approaches according to the number of strata used in the estimation of the censoring distribution in scenario~II.}
    \label{fig:sc2_var}
\end{figure}

As can be seen from Figure~\ref{fig:sc2_b}, all three approaches produce parameter estimates that on average resemble the average parameter estimates from the uncensored data when a large degree of stratification is used. The level of resemblance increases with the number of strata. At a low number of strata, the outcome weighting approach produces the worst resemblance and the pseudo-observation approach the best resemblance in this setting. This is in line with an earlier observation that the pseudo-observation approach may have the best chance of having a limit close to the best fit in uncensored data. 

The results on variance and variance estimation are presented in Figure~\ref{fig:sc2_var}. Notably, the outcome weighting approach produces a very large variance in the $\beta_2$ estimates, at least at a low number of strata, in comparison with the other approaches.
As suggested by an earlier observation, the variance in these parameter estimates for the outcome approach does decrease with the number of strata used in the estimation of the censoring distribution. As also suggested by an earlier observation, this decrease is not seen for corresponding standard sandwich variance estimates.
A similar pattern is seen for the individual weighting approach for the $\beta_2$ parameter, but at a much lower level of variance.
Overall, the pseudo-observation approach produces the smallest variances of parameter estimates and these observed variances seem well-reflected by corresponding average variance estimates. Perhaps surprisingly, the observed variances almost seem to tend to increase with the number of strata, but with the scale in mind, it must also be fair to claim that the observed variances do not change much with the number of strata.
There are similarly many apparent deviations from earlier theoretical and asymptotic observations: Observed variances increasing with the number of strata for the individual and outcome weighting approaches; variance estimates that are on average smaller than the observed variances. It is worth noting that some of these theoretical observations were made under the assumption of a true regression model and independent censoring in strata, which does not hold in this example.

In conclusion, this scenario demonstrates a setting with a total failure of the outcome weighting approach in terms of variance and also the failure of the standard sandwich variance estimator, at least for the individual and outcome weighting approaches.

\paragraph{Scenario III - risk ratios and stratifications in a factorial design.}
A $2^5$ factorial design is considered in this scenario with a complete symmetry in the 5 factors. Focus will be on estimation of the risk ratio related to one factor while taking the other factors into account.
The purpose is a comparison of the three approaches in low sample sizes and with a potentially heavy degree of stratification.
There is a potential for $2^5 = 32$ strata. Let a stratum be given by $x = (x_1, x_2, \dots, x_5)$ where $x_1, \dots, x_5 \in \{0,1\}$ denote the 5 factors. In such a given stratum, $T$ is drawn from a uniform distribution on $(0, 1/(0.1 \cdot 1.25^{x_1 + x_2 + \cdots + x_5}))$, and the outcome $Y = \indic{T \leq 1}$ is considered. 
The expectation is $\E(Y \given X=x) = 0.1 \cdot 1.25^{x_1 + \cdots + x_5}$.
For the estimation procedure, consider the model specified by $\mu(\beta; x) = \exp(\beta_0 + \beta_1 \cdot x_1 + \cdots + \beta_5 \cdot x_5)$. This means the true values have $\exp(\beta_j) = 1.25$ for $j=1, \dots, 5$. The parameter $\beta_1$ is considered of primary interest.
The censoring time is drawn, independently from $T$, from a uniform distribution on $(0,5/3)$. 
Scenarios with 2, 6, and 12 observations per stratum are considered. In other words, $n \in \{64, 192, 384\}$.
The three approaches are used with $A(\beta; X) = \frac{\partial}{\partial \beta} \mu(\beta; X)$ and stratification on either no factors, the first factor, the first three factors, and on all five factors.
Note that the leave-one-out calculations of the pseudo-observation approach are carried out at their lower limit of sample size in the case of stratification on all factors and only 2 observations per stratum. 

\begin{table}[htbp]
\begin{adjustwidth}{-6em}{-6em}
\footnotesize
\centering
\begin{tabular}{rrrrrrrrrrrrrrrrr}
  \toprule
$k$ & $n/32$ & $pc_{ind}$ & $pc_{out}$ & $pc_{pse}$ & $\%_{ind}$ & $\%_{out}$ & $\%_{pse}$ & $\widehat{\Var}_{ind}$ & $\widehat{\Var}_{out}$ & $\widehat{\Var}_{pse}$ & $\Var_{ind}$ & $\Var_{out}$ & $\Var_{pse}$ \\ 
  \midrule
  0 & 2 & 73.9 & 89.5 & 91.4 & 81.0 & 84.9 & 89.9 & 71.5 & 97.0 & 98.8 & 166.2 & 135.8 & 117.2 \\ 
  0 & 6 & 99.7 & 100.0 & 99.9 & 94.0 & 97.7 & 97.7 & 70.9 & 69.7 & 63.0 & 73.6 & 55.6 & 51.3 \\ 
  0 & 12 & 100.0 & 100.0 & 100.0 & 96.7 & 97.2 & 97.2 & 49.5 & 46.8 & 43.5 & 44.7 & 38.6 & 35.7 \\ 
  1 & 2 & 73.6 & 89.7 & 91.4 & 80.5 & 83.7 & 90.3 & 71.5 & 100.0 & 103.7 & 171.4 & 145.3 & 127.2 \\ 
  1 & 6 & 99.8 & 100.0 & 100.0 & 94.9 & 98.1 & 97.8 & 71.0 & 70.7 & 64.1 & 71.0 & 54.7 & 51.6 \\ 
  1 & 12 & 100.0 & 100.0 & 100.0 & 97.1 & 97.3 & 97.0 & 50.1 & 47.0 & 43.7 & 45.3 & 39.4 & 38.5 \\ 
  3 & 2 & 73.1 & 89.0 & 90.6 & 80.9 & 82.6 & 86.6 & 72.6 & 102.7 & 104.2 & 169.4 & 174.4 & 173.7 \\ 
  3 & 6 & 99.6 & 99.9 & 99.9 & 94.9 & 97.9 & 97.7 & 72.2 & 75.6 & 70.2 & 70.4 & 58.8 & 57.8 \\ 
  3 & 12 & 100.0 & 100.0 & 100.0 & 97.2 & 97.3 & 96.9 & 49.6 & 47.9 & 44.9 & 42.8 & 39.3 & 39.8 \\ 
  5 & 2 & 77.3 & 88.6 & 81.8 & 83.8 & 87.3 & 73.0 & 55.3 & 105.0 & 66.1 & 92.7 & 123.8 & 291.7 \\ 
  5 & 6 & 99.6 & 99.9 & 100.0 & 93.0 & 98.0 & 89.3 & 77.5 & 92.5 & 106.8 & 95.7 & 78.0 & 129.9 \\ 
  5 & 12 & 100.0 & 100.0 & 100.0 & 96.8 & 97.5 & 96.7 & 52.7 & 51.6 & 54.4 & 45.1 & 42.9 & 43.7 \\ 
   \bottomrule
\end{tabular}
\end{adjustwidth}
\caption{Convergence percentage ($pc$), coverage probability of Wald-type \SI{95}{\percent} confidence intervals based on the standard sandwich variance estimate (\%), median scaled standard sandwich variance estimate ($\widehat{\Var}$), and a scaled variance expression based on the  median absolute deviation of parameter estimates ($\Var$) for each of the three approaches for each of the configurations of used strata in censoring distribution estimation ($k$) and number of observations per strata ($n/32$).} 
\label{tab:sc3}
\end{table}

The results of Scenario~III are summarized in Table~\ref{tab:sc3}. Due to possible non-convergence for this estimation problem at small sample sizes, the convergence percentage is reported. The remaining statistics concern the select replications where convergence is achieved for all three approaches. Convergence is required to happen in 20 iterations. Due to possible outliers of parameter estimates and variance estimates, the median variance estimate scaled by $n$ and the scaled variance expression $\frac{n}{\Phi^{-1}(3/4)^2} \mad^2$ based on the median absolute deviation of parameter estimates, $\mad$, and standard normal cumulative distribution function, $\Phi$, are presented as robust alternatives to observed means and variances.

The individual weighting approach is seen to have the most convergence problems in small sample sizes in this setting. 
The coverage probabilities are generally too low for the three approaches at the small sample sizes, while they are slightly too large at the largest sample size. 
The pseudo-observation approach seems to have the largest problems in terms of coverage with a large degree of stratification at lower sample sizes. 
As suggested by the coverage probabilities, the summary statistics for variance and variance estimation can be quite off in settings with small sample sizes. 
At the largest sample size, the variance estimate summary is larger than the variance expression for the observed parameter estimates, which is in line with the theoretical observations made earlier.
At the largest sample size and the lower degree of stratification, the pseudo-observation approach is producing the lowest numbers in this setting, whereas the outcome weighting approach produces the lowest numbers at a larger degree of stratification.
Overall, nothing seems to be gained in terms of variance for either of the three approaches by applying a larger degree of stratification in this setting.

As a conclusion, this scenario has certainly crossed the line into territory where neither of these approaches work well. The scenario gives an example where the individual weighting approach seems more likely not to achieve convergence, and an example where the pseudo-observation approach is more vulnerable to too much stratification in the estimation of the censoring distribution. It also gives an example of a setting where not much is gained by applying stratification.

\section{Discussion}
\label{sec:discussion}

This paper has extended some of the results presented in Section~2 of \cite{blanche2022logistic}: In addition to the individual and outcome weighting approaches, called IPCW-GLM and OIPCW by \cite{blanche2022logistic} respectively, the pseudo-observation approach is now also considered; the type of outcome and type of model is now more general and not restricted to the logistic model; more results on the consequences of stratification in the estimation of the censoring distribution are given; and the presented results do not depend on continuity of the involved distributions, which may be useful when emulating expressions empirically.
A main point is also an extension of a result from \cite{blanche2022logistic}: Neither of the three approaches will generally have the lowest asymptotic variance. Which of the approaches will produce the lowest variance of parameter estimates will depend on the setting. The expressions given in this paper should help judging which one will win out.
Theoretical observations and simulations in this paper suggest the pseudo-observation approach may tend to have the most advantages such as smaller variance of parameter estimates and smaller bias from a misspecified censoring distribution. 
If it is of interest to use the standard sandwich variance estimator, this variance estimator also seems to be most appropriate for the pseudo-observation approach.
On the other hand, the pseudo-observation approach can be expected to be the most time-consuming approach, and it has also been seen to work less well when only few observations are available per stratum. 
Additionally, the simulation settings used in this paper may favor the pseudo-observation approach by design by having a larger censoring rate, or rather $ \Lambda(\d s)/G(s)$, at later time points rather than earlier time points. This may be realistic, but is of course not guaranteed in applications. 
One simulation scenario saw the outcome weighting approach producing a tremendously large variance in parameter estimates in a setting concerning restricted event time, $T \wedge t$. The theory would suggest this happens because $\Var(B(X) (T \wedge t) \given T > s, Z)$ is large for $s$ close to $t$. Arguably, this is a problem with the outcome definition rather than the outcome weighting approach. A similar model could be studied by focusing on, and weighting, the outcome $t - T \wedge t$, the time lost before $t$. The model could, in principle, be changed along the same lines to $\tilde \mu(\beta; X) = t - \mu(\beta; X)$ without changing the uncensored problem. With this outcome, $\Var(B(X) (t - T \wedge t) \given T > s, Z)$ would be small for $s$ close to $t$, and is expected to resolve the issue although this was not pursued here. 
This issue did not arise in \cite{Scheike2008} where focus was on risk, but it may be worthwhile having in mind when using the outcome weighting approach with other types of outcomes.

A theme of this paper has been the bias of the standard sandwich variance estimator as seen in Theorem~\ref{thm:varest}. The bias is upwards under the assumptions and leads to conservative inference if confidence intervals and statistical tests are based on this variance estimator. This applies to all three approaches. 
In simulation scenario~II, the variance estimate was seen to be much too large for the outcome weighting approach in some settings. In the paper by \cite{blanche2022logistic}, a simulation setting reveals considerable bias of the variance estimator in the individual weighting approach. 
The bias of the variance estimator in the pseudo-observation approach was studied by \cite{Overgaard2018} where the bias was found to be large in extreme cases, but tolerable in many less extreme cases. This seems to be in line with the simulation results of this paper.
Generally, it should be possible to construct a more asymptotically appropriate variance estimate by emulating the asymptotic variance expression empirically. Suggestions along these lines are given in Appendix~A.2 of \cite{blanche2022logistic} for the individual and outcome weighting approaches and in the end of section~3 of \cite{Overgaard2017} for the pseudo-observation approach. These variance estimators have not been studied in this paper.

The three approaches studied in this paper are rather simple and should not be expected to be efficient with the exception of some simple cases. 
A possible refinement is to consider an actual model of the censoring distribution and using fits from such a model in the weight. This is in fact considered by for instance \cite{Robins1992} for the individual weighting approach, by \cite{Scheike2008} for the outcome weighting approach, and by \cite{Overgaard2019} for the pseudo-observation approach. By better fitting the censoring distribution, a reduction in asymptotic variance of the regression parameter estimates can be expected.
Such approaches may also be augmented by including additional terms in the estimating equation that are designed to reduce the asymptotic variance. Including such terms tend to require a working model for the event time and type as well. Such an approach is also considered by \cite{blanche2022logistic}, citing \cite{Robins1992} and \cite{Bang2000}.
In \cite{Martinussen2023} an approach is considered where the the observed data efficient
influence function is emulated directly in a setting concerning risk regression.
Recently, an approach based on so-called censoring unbiased transformations was suggested by \cite{sandqvist2024doubly}, and this approach can in this context best be seen as a refinement of the pseudo-observation approach using working models of the censoring and the outcome to calculate an improved pseudo-observation while obtaining double robustness and oracle efficiency properties.
It has been beyond the scope of this paper to study these approaches more closely.
In contrast, the focus has here been on simple methods that can be applied with access only to standard tools for statistical analysis, such as tools for the Kaplan--Meier estimator, estimation of parameters in a generalized linear model, and the standard sandwich variance estimator.
The assumptions considered do not seem overly restrictive. The most restrictive assumptions are likely the ones imposed on the censoring mechanism. These assumptions can be made less restrictive by appealing to stratification in the censoring distribution estimation. 
It would be of interest to study whether a considerable amount of efficiency can in fact be gained by applying the more advanced approaches mentioned above, or if it is possible in many settings to achieve a reasonable amount of efficiency by applying an appropriate amount of stratification.

\section*{Acknowledgements}
The author would like to thank Jan Pedersen for helpful discussions on the subject and insightful comments on drafts of this paper.

\appendix

\section{Technical results in a broader setting}
\label{sec:appendix}

In this appendix, some technical results that will help to establish the primary results of this paper are considered. 
Sufficient regularity conditions, including positivity, are imposed, but it is of interest not to invoke the independent censoring assumption at this stage. 
This means that the estimators considered may not be consistent for their intended estimand. This is specifically the case for the estimates of $\Lambda$ and $G$, related to the conditional censoring distribution. As is made clear below, the modified Nelson--Aalen estimate of $\Lambda(\cdot \given z)$ uses empirical estimates of $\check S(s \given z) = \P(T > s, C \geq s \given Z = z)$ and $\tilde F_0(s \given z) = \P(\tilde T \leq s, \tilde D = 0 \given Z = z)$, and the limits of the estimates of $\Lambda$ and $G$ are therefore instead 
\begin{align}
\Lambda^*(s \given z) &= \int_0^s \frac{1}{\check S(u \given z)} \tilde F_0(\d u \given z), \\
G^*(s \given z) & = \Prodi_0^s(1-\Lambda^*(\d u \given z)),
\end{align}
at least under a condition of positivity, $\check S(s \given z) > 0$ for (almost) all $z$. 
In the following, the notation $\mathcal{X} = (T, D, Z, X)$ is used for the underlying information, $\mathcal{\tilde X} = (\tilde T, \tilde D, Z)$ is used for the information used in estimation of the weights, and $\mathcal{\bar X} = (\tilde T, \tilde D, Z, X)$ is used for the observed information used in the estimating equation. Also, let
\begin{equation}
M(\mathcal{\tilde X}; s \given z) = N(\mathcal{\tilde X}; s; z) - \int_0^s R(\mathcal{\tilde X}; u; z) \Lambda^*(\d u \given z),
\end{equation}
where
\begin{align}
N(\mathcal{\tilde X}; s; z) &= \indic{\tilde T \leq s, \tilde D = 0, Z = z}, \\
R(\mathcal{\tilde X}; s; z) &= \indic{\tilde T > s, Z = z} + \indic{\tilde T = s, \tilde D=0, Z=z}.
\end{align}
The notation $p_Z(z) = \P(Z=z)$ is used for the distribution of $Z$ and $\tilde S(s \given z) = \P(\tilde T > s \given Z=z)$ is used for the conditional distribution of $\tilde T$. It may be noted that $\tilde S(s \given z) = \check S(s \given z) (1 - \Delta \Lambda^*(s \given z))$ without invoking further assumptions.
To be more explicit on the estimation of the weights, consider data on $n$ replications and let $\bar N(s; z) = \frac{1}{n} \sum_{i=1}^n N(\mathcal{\tilde X}_i; s; z)$ and $\bar R(s; z) = \frac{1}{n} \sum_{i=1}^n R(\mathcal{\tilde X}_i; s; z)$. Then the estimates of the cumulative censoring hazard and the censoring survival function are given by
\begin{align}
\hat \Lambda(s \given z) = \int_0^s \frac{1}{\bar R(u; z)} \bar N(\d u; z), \\
\hat G(s \given z) = \Prodi_0^s (1 - \hat \Lambda(\d u \given z) ), \label{eq:G_est}
\end{align}
A useful result for approximating the applied weights is the following.

\begin{lemma} \label{lemma:G_approximation}
For a given $z$ with $p_Z(z) > 0$, it holds, for any $s$ such that $\tilde S(s \given z) > 0$, that 
\begin{equation}
\sup_{u \in [0,s]} \big|\frac{1}{\hat G(u \given z)} - \frac{1}{G^*(u \given z)} + \frac{1}{n} \sum_{i=1}^n \frac{1}{G^*(u \given z)} \frac{\dot G(\mathcal{\tilde X}_i; u \given z)}{G^*(u \given z)}\big| = o_{\P}(n^{-1/2})
\end{equation}
where
\begin{equation}
\dot G(\mathcal{\tilde X}; u \given z) = - G^*(u \given z) \int_0^u \frac{1}{\tilde S(v \given z) p_Z(z)} M(\mathcal{\tilde X}; \d v \given z)
\end{equation}
\end{lemma}
\begin{proof}
One way of proving this result is by taking a functional approach similar to the approach of \cite{Overgaard2019}. Viewing the estimator as an application of a functional between spaces of functions of bounded $p$-variation to the empirical distribution of the data on $\mathcal{\tilde X}$, the result is obtained by finding the derivative of the functional. 

For any distribution function, $F$, of $\mathcal{\tilde X}$, at risk and censoring indicating functionals defined by
\begin{align*}
R(F; s; z) &= \int \indic{T > s, C \geq s, Z = z} \d F, \\
N(F; s; z) &= \int \indic{\tilde T \leq s, \tilde D = 0, Z = z} \d F,
\end{align*}
that take the expectation with respect to that distribution, can be considered. Next, a corresponding cumulative censoring hazard can be defined by $\Lambda(F; s \given z) = \int_0^s \frac{1}{R(F; u; z)} N(F; \d u ; z)$, and the corresponding censoring distribution is given by $G(F; s \given z) = \prodi_0^s(1-\Lambda(F; \d u \given z))$. Finally, a functional defined by 
\begin{equation}
    \phi(F; s \given z) = \frac{1}{G(F; s \given z)}
\end{equation}
will result in $1/\hat G( \cdot \given z)$ when applied to the empirical distribution, and it may be noted that $1/G^*( \cdot \given z)$ is obtained when $\phi$ is applied to the true distribution. 
A derivative on the form
\begin{equation}
    \phi_F'(f; s \given z) = - \frac{1}{G(F; s \given z)} \frac{G_F'(f; s \given z)}{G(F; s \given z)}
\end{equation}
is expected, where, according to differentiability results on the product integral,
\begin{equation}
    G_F'(f; s \given z) = - G(F; s \given z) \int_0^s \frac{1}{1- \Delta \Lambda(F; u \given z)} \Lambda_F'(f; \d u \given z)
\end{equation}
and, owing to bilinearity of the integral and linearity of $R$ and  $N$,
\begin{equation}
    \Lambda_F'(f; s \given z) = \int_0^s \frac{1}{R(F; u; z)} N(f; \d u; z) - \int_0^s \frac{R(f; u; z)}{R(F; u; z)^2} N(F; \d u \given z).
\end{equation}
With these expressions in hand, it can be seen that the result follows if it can be argued that $\phi(F_n; \cdot \given z)$ is sufficiently close to $\phi(F_0; \cdot \given z) + \phi_{F_0}'(F_n - F_0; \cdot \given z)$ for large $n$ where $F_n$ is the empirical distribution and $F_0$ the true distribution.
Note for instance how $\Lambda_{F_0}'(F_n - F_0; s \given z) = \frac{1}{n} \sum_i \int_0^s \frac{1}{R(F_0;u;z)} M(\mathcal{\tilde X}_i; \d u \given z)$ and $(1- \Delta \Lambda(F_0; u \given z))R(F_0;u;z) = (1- \Delta \Lambda^*(u \given z)) \check S(u \given z) p_Z(z) = \tilde S(u \given z) p_Z(z)$.
Under the stated assumptions, the convergence result in $p$-variation over $[0,s]$
\begin{equation}
    \|\phi(F_n; \cdot \given z) - \phi(F_0; \cdot \given z) - \phi_{F_0}'(F_n - F_0; \cdot \given z)\|_{[p]} = O_{\P}(n^{2\frac{1-p}{p}})
\end{equation}
holds for $p \in (1, 2)$ and so the same convergence order holds in supremum norm over $[0,s]$.
The result is based on the convergence order $\|F_n - F_0\| = O_{\P}(n^{(1-p)/p})$ in a norm based on $p$-variation and that $\phi$ is more than once continuously differentiable, yielding a first order remainder of order $O(\|F_n - F_0\|^2)$ as $F_n$ approaches $F_0$.
The required convergence order is obtained for $p > 4/3$, but faster convergence orders, almost $O_{\P}(n^{-1})$, can also be obtained by this argument by considering $p$ close to 2.
See \cite{Overgaard2019} for further details on the choice of norm and the differentiability results.
\end{proof}

\begin{lemma} \label{lemma:influence_theta}
Assuming $\tilde S(t \given Z) > 0$ almost surely, the estimator 
\begin{equation}
\hat \theta = \frac{1}{n} \sum_{i=1}^n \hat W_i Y_i
\end{equation}
has first order influence function 
\begin{equation}
\dot \theta(\mathcal{\tilde X}) = W^*Y - \E(W^* Y) + \int_0^{t-} e^*(u \given Z) M(\mathcal{\tilde X}; \d u \given Z) 
\end{equation}
and second order influence function
\begin{equation}
\begin{aligned}
\ddot \theta(\mathcal{\tilde X}_1, \mathcal{\tilde X}_2) &= \int_0^{t-} \big(W_1^* Y_1 \frac{\indic{\tilde T_1 > s} \indic{Z_1=Z_2}}{\tilde S(s \given Z_2) p_Z(Z_2)} \\
& \hspace{-2em} - e^*(s \given Z_2) \big(\frac{R(\mathcal{\tilde X}_1; s; Z_2)}{\check S(s \given Z_2) p_Z(Z_2)} - \dot \Lambda(\mathcal{\tilde X}_1; s \given Z_2) \big) \big) M(\mathcal{\tilde X}_2; \d s \given Z_2) \\
& \phantom{{}=} + \int_0^{t-} \big(W_2^* Y_2 \frac{\indic{\tilde T_2 > s} \indic{Z_2=Z_1}}{\tilde S(s \given Z_1) p_Z(Z_1)} \\
& \hspace{-2em} - e^*(s \given Z_1) \big(\frac{R(\mathcal{\tilde X}_2; s; Z_1)}{\check S(s \given Z_1) p_Z(Z_1)} - \dot \Lambda(\mathcal{\tilde X}_2; s \given Z_1) \big) \big) M(\mathcal{\tilde X}_1; \d s \given Z_1) \\
\end{aligned}
\end{equation}
where $W^* = \indic{C \geq T \wedge t}/G^*(T \wedge t- \given Z)$ and $e^*(s \given z) = \E(W^*Y \given \tilde T > s, Z = z)$ and
\begin{equation} \label{eq:dotLambda}
\dot \Lambda(\mathcal{\tilde X}; s \given z) = \int_0^s \frac{1}{\tilde S(u \given z) p_Z(z)}M(\mathcal{\tilde X}; \d s \given z).    
\end{equation}

\end{lemma}

\begin{proof}
In similarity to the proof of Lemma~\ref{lemma:G_approximation}, it is possible to use a functional approach, and the functionals defined in that proof are reused here. Specifically, it is possible to see the estimator as an evaluation of a functional 
\begin{equation}
\theta: F \mapsto \int W(F) Y \d F = \int \frac{\indic{C \geq T \wedge t}}{G(F; \tilde T \wedge t- \given Z)} Y \d F
\end{equation}
at the empirical distribution of $\mathcal{\tilde X}$. 
The first order derivative of the functional $\theta$ is given by
\begin{equation}
\theta_F'(f) = \int W(F) Y \d f + \int W_F'(f) Y \d F
\end{equation}
while the second order derivative is given by
\begin{equation}
\theta_F''(f, g) = \int W_F'(g) Y \d f + \int W_F'(f) Y \d g + \int W_F''(f, g) Y \d F.
\end{equation}
Here,
\begin{equation}
W_F'(f) = -\frac{\indic{C \geq T \wedge t}}{G(F; \tilde T \wedge t- \given Z)} \frac{G_F'(f; \tilde T \wedge t- \given Z)}{G(F; \tilde T \wedge t- \given Z)}
\end{equation}
and
\begin{equation}
\begin{aligned}
W_F''(f,g) &= 2 \frac{\indic{C \geq T \wedge t}}{G(F; \tilde T \wedge t- \given Z)} \frac{G_F'(f; \tilde T \wedge t- \given Z)}{G(F; \tilde T \wedge t- \given Z)} \frac{G_F'(g; \tilde T \wedge t- \given Z)}{G(F; \tilde T \wedge t- \given Z)} \\
& \phantom{{}=} - \frac{\indic{C \geq T \wedge t}}{G(F; \tilde T \wedge t- \given Z)} \frac{G_F''(f, g; \tilde T \wedge t- \given Z)}{G(F; \tilde T \wedge t- \given Z)}.
\end{aligned}
\end{equation}
At this stage it may be worthwhile to introduce a functional by
\begin{equation}
\Gamma_F(f; s \given z) = -\frac{G_F'(f; s \given z)}{G(F; s \given z)} = \int_0^s \frac{1}{1- \Delta \Lambda(F; u \given z)} \Lambda_F'(f; \d u \given z).
\end{equation}
Note that 
\begin{equation}
\Lambda_F''(f,g; s \given z) = - \int_0^s \frac{R(g; u; z)}{R(F; u; z)} \Lambda_F'(f; \d u \given z) - \int_0^s \frac{R(f; u; z)}{R(F; u; z)} \Lambda_F'(g; \d u \given z)
\end{equation}
such that, in a short notation,
\begin{equation}
\begin{aligned}
G_F''(f,g; s \given z) &= \int_0^s \int_0^{u-} \prodi_0^{v-} (1- \d \Lambda) \Lambda_F'(g ; \d v \given z) \prodi_v^{u-} (1 - \d \Lambda) \Lambda_F'(f; \d u \given z) \prodi_u^s(1- \d \Lambda) \\
& \phantom{{}=} + \int_0^s \prodi_0^{u-} (1- \d \Lambda) \Lambda_F'(f ; \d u \given z) \int_u^s \prodi_u^{v-} (1 - \d \Lambda) \Lambda_F'(g; \d v \given z) \prodi_v^s(1- \d \Lambda) \\
& \phantom{{}=} - \int_0^s \prodi_0^{u-} (1- \d \Lambda) \Lambda_F''(f,g; \d u \given z) \prodi_u^s (1 - \d \Lambda) \\
&= G(F; s \given z) \int_0^s (\Gamma_F(g; u- \given z)+\frac{R(g; u; z)}{R(F; u; z)}) \Gamma_F(f; \d u \given z) \\
& \phantom{{}=} + G(F; s \given z)\int_0^s (\Gamma_F(f; u- \given z)+\frac{R(f; u; z)}{R(F; u; z)}) \Gamma_F(g; \d u \given z).
\end{aligned}
\end{equation}
Using that $\Gamma_F(f; \tilde T \wedge t- \given z) - \Gamma_F(f; u- \given z) = \int_{u-}^{\tilde T \wedge t-} \Gamma_F(f; \d v \given z)$ and changing the order of integration, the expression
\begin{equation}
\begin{aligned}
W_F''(f,g) &= \int_0^{t-} W(F) \indic{\tilde T > u} (\Gamma_F(g; u \given Z) - \frac{R(g; u; Z)}{R(F; u; Z)}) \Gamma_F(f; \d u \given Z) \\
& + \int_0^{t-} W(F) \indic{\tilde T > u} (\Gamma_F(f; u \given Z) - \frac{R(f; u; Z)}{R(F; u; Z)}) \Gamma_F(g; \d u \given Z)
\end{aligned}
\end{equation}
is obtained.
Also, in the same terms,
\begin{equation}
W_F'(f) = \int_0^{t-} W(F) \indic{\tilde T > u } \Gamma_F(f; \d u \given Z).
\end{equation}
This leads to the expressions
\begin{equation}
\theta_F'(f) =\int W(F) Y \d f + \int \int_0^{t-} W(F) Y \indic{\tilde T > u} \Gamma_F(f; \d u \given Z) \d F
\end{equation}
and
\begin{equation}
\begin{aligned}
\theta_F''(f,g) &= \int \int_0^{t-} W(F) Y \indic{\tilde T > u} \Gamma_F(g; \d u \given Z) \d f \\
& \phantom{{}=} + \int \int_0^{t-} W(F) Y \indic{\tilde T > u} \Gamma_F(f; \d u \given Z) \d g \\
& \phantom{{}=} + \int \int_0^{t-} W(F) Y \indic{\tilde T > u} (\Gamma_F(g; u \given Z) - \frac{R(g; u; Z)}{R(F; u; Z)}) \Gamma_F(f; \d u \given Z) \d F \\
& \phantom{{}=} + \int \int_0^{t-} W(F) Y \indic{\tilde T > u} (\Gamma_F(f; u \given Z) - \frac{R(f; u; Z)}{R(F; u; Z)}) \Gamma_F(g; \d u \given Z) \d F.
\end{aligned}
\end{equation}
To obtain the first order influence function, evaluate the first order derivative at the true $F=F_0$ and the direction $f = \delta_{\mathcal{\tilde X}} - F_0$ where $\delta_{\mathcal{\tilde X}}$ corresponds to the Dirac measure at a given $\mathcal{\tilde X}$. Note that (for $z$ such that $p_Z(z) > 0$) 
\begin{equation}
\Gamma_{F_0}(\delta_{\mathcal{\tilde X}} - F_0; s \given z) = \int_0^s \frac{\indic{Z = z}}{\tilde S(u \given z) p_Z(z)} M(\mathcal{\tilde X}; \d u \given Z).
\end{equation}
Since $Z$ and $z$ can be used interchangeably on $Z=z$, this allows for a useful change in the order of integration such that, since
\begin{equation}
e^*(s \given z) = \int W(F_0) Y \frac{\indic{\tilde T > s, Z = z}}{\tilde S(s \given z) p_Z(z)} \d F_0,
\end{equation}
the desired expression of the first order influence function is obtained. 
To obtain the second order influence function, evaluate the second order derivative at the true $F=F_0$ and directions $f=\delta_{\mathcal{\tilde X}_1} - F_0$ and $g=\delta_{\mathcal{\tilde X}_2} - F_0$. Using the same arguments as for the first order influence function and with a slight elimination of terms, the desired expression is obtained.

The functional approach used here can be formalized in a $p$-variation setting as is done in \cite{Overgaard2019}, but this does impose restrictive requirements on the outcome $Y$, namely bounded $p$-variation on the time argument of the function $y$ from section~\ref{sec:regression_analysis}. The implied influence functions likely apply more generally. 
\end{proof}

The notation in \eqref{eq:dotLambda} is slightly misleading since $\dot \Lambda$ is not generally the exact influence function of the $\Lambda$ estimate. The close connection of $\dot \Lambda$ to the influence function can however be realized from the proof. In the continuous case there is no difference.

Recall that the estimating equation for each type is of the form $U_{n,\textup{type}}(\beta) = 0$ for $U_{n,\textup{type}}(\beta) = \sum_i u_{n,i,\textup{type}}(\mathcal{\bar X}_i;\beta)$ where specifically
\begin{align}
    u_{n,i,\textit{ind}}(\mathcal{\bar X}_i;\beta) & = A(\beta; X_i) \hat W_i (Y_i - \mu(\beta; X_i)), \\
    u_{n,i,\textit{out}}(\mathcal{\bar X}_i;\beta) & = A(\beta; X_i) (\hat W_i Y_i - \mu(\beta; X_i)), \\
    u_{n,i,\textit{pse}}(\mathcal{\bar X}_i;\beta) & = A(\beta; X_i) (\hat \theta_i - \mu(\beta; X_i)).
\end{align}
The next lemma concerns an approximation of these terms.
In the remainder of this appendix, a set of regularity conditions are used to establish the desired properties.
A sufficient set of conditions involve: The positivity requirement $\tilde S(t \given Z) > 0$ (almost surely); a bounded outcome $Y$; two times continuous differentiability of $\beta \mapsto A(\beta; X)$ and $\beta \mapsto A(\beta; X) \mu(\beta; X)$ (almost surely); locally dominated integrability of the second order derivatives of these functions; integrability of the first order derivative; finite second moment of $A(\beta; X)$ and $A(\beta; X) \mu(\beta;X)$; non-singularity of relevant matrices in the following. Many of these conditions are to hold at or in an open neighborhood of the relevant $\beta$. 

\begin{lemma} \label{lemma:contribution_approximation}
Under the the mentioned regularity conditions each of the three types allows for a representation
\begin{equation}
u_{n, i,\textup{type}}(\beta) = u_{\textup{type}}(\mathcal{\bar X}_i; \beta) + \frac{1}{n} \sum_{j=1}^n \dot u_{\textup{type}}(\mathcal{\bar X}_i; \mathcal{\bar X}_j; \beta)+ R_{n, i, \textup{type}}(\beta)
\end{equation}
where 
\begin{enumerate}
\item $\E(\dot u_{\textup{type}}(\mathcal{\bar X}_i; \mathcal{\bar X}_j; \beta) \given \mathcal{\bar X}_i) = 0$ (almost surely) for $j \neq i$. 
\item $\| R_{n,i, \textup{type}}(\beta) \| \leq r_\textup{type}(\mathcal{\bar X}_i;\beta) Q_{n,\textup{type}}$ for a positive and locally dominated square integrable function $r_\textup{type}(\mathcal{\bar X}_i; \beta)$ and a sequence $(Q_{n,\textup{type}})$ such that $Q_{n,\textup{type}} = o_{\P}(n^{-1/2})$ as $n \to \infty$. 
\end{enumerate}
Specifically, 
\begin{align*}
u_\text{ind}(\mathcal{\bar X}; \beta) &= A(\beta; X) W^*(Y- \mu(\beta;X)), \\
u_\text{out}(\mathcal{\bar X}; \beta) &= A(\beta; X) (W^*Y- \mu(\beta;X)), \\
u_\text{pse}(\mathcal{\bar X}; \beta) &= A(\beta; X) (W^*Y - \mu(\beta;X) + \int_0^{t-} e^*(u \given Z)  M(\mathcal{\tilde X}; \d u \given Z)),
\end{align*}
and
\begin{align*}
\dot u_\text{ind}(\mathcal{\bar X}_1, \mathcal{\bar X}_2;\beta) &= A(\beta; X_1) \int_0^{t-} W_1^*(Y_1 - \mu(\beta; X_1)) \frac{\indic{\tilde T_1 > u} \indic{Z_1=Z_2}}{\tilde S(u \given Z_2) p_Z(Z_2)} M(\mathcal{\tilde X}_2; \d u \given Z_2), \\
\dot u_\text{out}(\mathcal{\bar X}_1, \mathcal{\bar X}_2;\beta) &= A(\beta; X_1) \int_0^{t-} W_1^* Y_1 \frac{\indic{\tilde T_1 > u} \indic{Z_1=Z_2}}{\tilde S(u \given Z_2) p_Z(Z_2)} M(\mathcal{\tilde X}_2; \d u \given Z_2), \\
\dot u_\text{pse}(\mathcal{\bar X}_1, \mathcal{\bar X}_2;\beta) &= A(\beta; X_1) \ddot \theta(\mathcal{\tilde X}_1, \mathcal{\tilde X}_2).
\end{align*}
\end{lemma}
\begin{proof}
For types \textit{ind} and \textit{out} this is an application of Lemma~\ref{lemma:G_approximation} where $r_\textit{ind}(\mathcal{\bar X}; \beta) = \|A(\beta;X)(Y-\mu(\beta;X)) \indic{C \geq T \wedge t}\|$ and $r_\textit{out}(\mathcal{\bar X}; \beta) = \|A(\beta;X) Y \indic{C \geq T \wedge t}\|$ can be used and where $Q_{n,\textup{type}}$ can be the remainder from Lemma~\ref{lemma:G_approximation} with $s=t$. 
Concerning the \textit{pse} type, the approximation of the pseudo-observations
\begin{equation} \label{eq:pseudo_approx}
\hat \theta_i = \theta + \dot \theta(\mathcal{\tilde X}_i) + \frac{1}{n} \sum_{j=1}^n \ddot \theta(\mathcal{\tilde X}_i, \mathcal{\tilde X}_j) + o_{\P}(n^{-1/2})
\end{equation}
is uniform in $i$ according to Proposition~3.1 of \cite{Overgaard2017} using the functional approach of Lemma~\ref{lemma:influence_theta}. Lemma~\ref{lemma:influence_theta} also gives the expression of $\dot \theta$ and $\ddot \theta$ used in the statement. Above, $\theta = \E(W^* Y)$ is the limit of the estimator. Concretely, take $r_\textit{pse}(\mathcal{\bar X}; \beta) = \|A(\beta;X) \|$ and let $Q_{n,\textit{pse}}$ be the remainder from \eqref{eq:pseudo_approx}.
The property that $\E(\dot u_{\textup{type}}(\mathcal{\bar X}_i; \mathcal{\bar X}_j; \beta) \given \mathcal{\bar X}_i) = 0$ for $j \neq i$ is essentially a result of properties of the influence functions, but can be checked using primarily that $\E(M(\mathcal{\tilde X}; s \given Z) \given Z) = 0$ as well as independence of observations. For the \textit{pse} case, it is worth noting that $\E(R(\mathcal{\tilde X};s; z)) = \check S(s \given z) p_Z(z)$ and that 
\begin{equation}
\E \big(W^* Y \frac{\indic{\tilde T > s} \indic{Z_1 = z}}{\tilde S(s \given z) p_Z(z)} \big) = e^*(s \given z)
\end{equation}
basically by definition.
\end{proof}

With the definitions of Lemma~\ref{lemma:contribution_approximation}, let for each of the three types
\begin{equation}
h_{\textup{type}}(\mathcal{\bar X}; \beta) = u_{\textup{type}}(\mathcal{\bar X}; \beta) + \E(\dot u_{\textup{type}}(\mathcal{\bar X}_1; \mathcal{\bar X}; \beta) \given \mathcal{\bar X}).
\end{equation}

\begin{lemma} \label{lemma:asymptotic_equivalence_estimating_equation}
Under the regularity conditions mentioned above there is, for each of the three types, an asymptotic equivalence of $U_{n,\textup{type}}(\beta) = \sum_i u_{n,i,\textup{type}}(\mathcal{\bar X}_i;\beta)$ and 
\begin{equation}
U^*_{n,\textup{type}}(\beta) = \sum_{i=1}^n h_{\textup{type}}(\mathcal{\bar X}_i; \beta)
\end{equation}
in the sense that $n^{-1/2}(U_{n,\textup{type}}(\beta)-U^*_{n,\textup{type}}(\beta)) \to 0$ in probability as $n \to \infty$.
\end{lemma}
\begin{proof}
This is a $U$-statistic or rather $V$-statistic argument, which applies owing to the approximations of Lemma~\ref{lemma:contribution_approximation}. The $U_{n,\textup{type}}(\beta)$ is well approximated by a $V$-statistic of order 2. Symmetrized, the $V$-statistic has the kernel function
\begin{equation}
\begin{aligned}
k(\mathcal{\bar X}_1, \mathcal{\bar X}_2; \beta) &= \frac{1}{2} \big(u_{\textup{type}}(\mathcal{\bar X}_1; \beta) + u_{\textup{type}}(\mathcal{\bar X}_2; \beta) \\
& \phantom{{}=} + \dot u_{\textup{type}}(\mathcal{\bar X}_1, \mathcal{\bar X}_2;\beta) + \dot u_{\textup{type}}(\mathcal{\bar X}_2, \mathcal{\bar X}_1;\beta)\big).
\end{aligned}
\end{equation}
Let $\gamma = \E(k(\mathcal{\bar X}, \mathcal{\bar X}_2; \beta))$. Using the properties from Lemma~\ref{lemma:contribution_approximation}, it can be seen that $\E(\dot u(\mathcal{\bar X}_1, \mathcal{\bar X}_2;\beta)) = 0$, such that $\gamma = \E(u_{\textup{type}}(\mathcal{\bar X}; \beta))$, and then
\begin{equation}
\begin{aligned}
2 ( \E(k(\mathcal{\bar X}, \mathcal{\bar X}_2; \beta) \given \mathcal{\bar X}) - \gamma) &= u_{\textup{type}}(\mathcal{\bar X}; \beta) + \E(\dot u_{\textup{type}}(\mathcal{\bar X}_1; \mathcal{\bar X}; \beta) \given \mathcal{\bar X}) - \gamma \\
&= h_{\textup{type}}(\mathcal{\bar X}; \beta) - \gamma.
\end{aligned}
\end{equation}
Since $k(\mathcal{\bar X}_1, \mathcal{\bar X}_2; \beta)$ will have finite second moment, as can be seen from the regularity conditions, the claim follows from results on $V$-statistics. See for instance Theorem~12.3 and Problem~12.10 of \cite{Vaart1998}.
\end{proof}

In the following lemma, the regularity conditions are used to ensure, with high probability, the existence of a solution to the estimating equation in each of the three cases. This can be done since the regularity conditions imply similar properties of $\beta \mapsto h_{\textup{type}}(\mathcal{\bar X}; \beta)$.

\begin{lemma} \label{lemma:estimator_exist}
For any one of the three types, suppose $\beta_{\textup{type}}^*$ exists such that $\E(h_{\textup{type}}(\mathcal{\bar X}; \beta_{\textup{type}}^*)) = 0$.
Under the regularity conditions mentioned above, a sequence $(\hat \beta_{n,\textup{type}})$ exists such that $U_{n,\textup{type}}(\hat \beta_{n,\textup{type}}) = 0$ with a probability tending to 1 and $\hat \beta_{n,\textup{type}} \to \beta_{\textup{type}}^*$ in probability, and further
\begin{equation}
\hat \beta_{n,\textup{type}} - \beta_{\textup{type}}^* = \frac{1}{n} \sum_{i=1}^n \dot \beta_{\textup{type}}(\mathcal{\bar X}_i) + o_{\P}(n^{-\frac{1}{2}})
\end{equation}
where 
\begin{equation}
\dot \beta_{\textup{type}}(\mathcal{\bar X}) = - J_{\textup{type}}(\beta_{\textup{type}}^*)^{-1} h_{\textup{type}}(\mathcal{\bar X}; \beta_{\textup{type}}^*)
\end{equation}
and
\begin{equation}
J_{\textup{type}}(\beta) = \E(\frac{\partial}{\partial \beta \T} h_{\textup{type}}(\mathcal{\bar X}; \beta)).
\end{equation}
\end{lemma}
\begin{proof}
On the basis of Lemma~\ref{lemma:asymptotic_equivalence_estimating_equation}, it is possible to appeal to Theorem~5.41 and Theorem~5.42 of \cite{Vaart1998} for this result. Strictly speaking, Theorem~5.42 will ensure the existence of a solution to $U_{n,\textup{type}}^*(\beta) = 0$ with high probability for large $n$, but an inspection of its proof will reveal that the same applies to a solution to $U_{n,\textup{type}}(\beta) = 0$ owing to the close approximation to the $V$-statistic considered in the proof of Lemma~\ref{lemma:asymptotic_equivalence_estimating_equation}, for which a uniform law of large numbers apply, and how well behaved the $r_{\textup{type}}(\mathcal{\bar X}_i; \beta)$ of the remainders are.
\end{proof}

The notation 
\begin{equation}
u(\mathcal{X};\beta) = A(\beta; X)(Y-\mu(\beta;X))
\end{equation}
relating to the estimating equation of the uncensored problem is used in the following.

\begin{proposition} \label{prop:bias}
In the setting of the previous lemma, suppose $\beta^*$ is a solution to the uncensored problem, $\E(u(\mathcal{X}; \beta)) = 0$ and $\beta_{\textup{type}}^*$ exists such that $\E(h_{\textup{type}}(\mathcal{\bar X}; \beta_{\textup{type}}^*)) = 0$ for one of the three types. Assuming invertibility of $J_{\textup{type}}$ at $\beta^*$, a first order approximation of the bias is 
\begin{equation}
\begin{aligned}
\beta^*_{\textup{type}} - \beta^* &\approx -J_{\textup{type}}(\beta^*)^{-1} \E(h_{\textup{type}}(\mathcal{\bar X}; \beta^*)) \\
&= \E(\int_0^{t-} \psi_{\textup{type}}(\mathcal{ X}; s; \beta^*) \frac{1}{G^*(s \given Z)} M(\mathcal{\tilde X}; \d s \given Z) )
\end{aligned}
\end{equation}
where, with the notation $B_{\textup{type}}(\beta; X) = J_\textup{type}(\beta)^{-1} A(\beta; X)$ and $W^*(s) = \frac{\indic{C \geq T \wedge t}}{G^*(T \wedge t \given Z)/G^*(s \given Z)}$,
\begin{align}
    \psi_\textit{ind}(\mathcal{ X}; s; \beta) &= B_{\textit{ind}}(\beta; X)(Y - \mu(\beta; X)), \\
    \psi_\textit{out}(\mathcal{ X}; s; \beta) &= B_{\textit{out}}(\beta; X)Y, \\
    \psi_\textit{pse}(\mathcal{ X}; s; \beta) &= B_{\textit{pse}}(\beta; X)(Y - \E(W^*(s) Y \given \tilde T > s, Z)). 
\end{align}
Under an assumption of $(T, D) \independent C \given X$, this leads to 
\begin{equation} \label{eq:bias}
\begin{aligned}
\beta^*_{\textup{type}} - \beta^* &\approx \E(\int_0^{t-} \E(\psi_{\textup{type}}(\mathcal{ X}; s; \beta^*) \given T > s, X) S(s \given X) \frac{G(s- \given X)}{G^*(s \given Z)} (\Lambda(\d s \given X) - \Lambda^*(\d s \given Z))).
\end{aligned}
\end{equation}
\end{proposition}
\begin{proof}
This result is based on a Taylor approximation of $\beta \mapsto \E(h_{\textup{type}}(\mathcal{\bar X}; \beta))$ around $\beta^*$ which would reveal 
\begin{equation}
0 = \E(h_{\textup{type}}(\mathcal{\bar X}; \beta^*_{\textup{type}})) = \E(h_{\textup{type}}(\mathcal{\bar X}; \beta^*)) + J_{\textup{type}}(\beta^*) (\beta_{\textup{type}}^* - \beta^*) + \textup{Rem}.
\end{equation}
and thereby the stated first order approximation. Here, $\E(h_{\textup{type}}(\mathcal{\bar X}; \beta^*)) = \E(u_{\textup{type}}(\mathcal{\bar X}; \beta^*) - u(\mathcal{X}; \beta^*))$ since the higher order term has mean 0 and using that $\beta^*$ is a solution to $\E(u(\mathcal{X}; \beta)) = 0$. The equality $W^*-1 = - \int_0^{T \wedge t-} \frac{1}{G^*(s \given Z)} M(\mathcal{\tilde X}; \d s \given Z)$ now helps to establish the second expression of the approximation.
The conditional independence assumption makes it possible to first take the expectation given $(T, D, X)$, which helps to establish the structure of the integrator in the display, and then conditional expectation given $X$ only. 
\end{proof}

\begin{lemma} \label{lemma:rep_independent}
Consider the setting of Lemma~\ref{lemma:estimator_exist}. Under the independent censoring assumption, Assumption~\ref{ass:indep_cens}, each of the three types allows for the representation 
\begin{equation} \label{eq:h_type_representation}
h_\textup{type}(\mathcal{\bar X};\beta) = u(\mathcal{X};\beta) - \int_0^{t-} \frac{v_\textup{type}(\mathcal{X}; s; \beta) - e_\textup{type}(\beta; s \given Z)}{G(s \given Z)} M(\mathcal{\tilde X}; \d s \given Z)
\end{equation}
where 
\begin{align}
v_\text{ind}(\mathcal{X};s;\beta) &= A(\beta; X)(Y - \mu(\beta; X)), \\
v_\text{out}(\mathcal{X};s;\beta) &= A(\beta; X)Y, \\
v_\text{pse}(\mathcal{X};s;\beta) &= A(\beta; X)(Y - \E(Y \given T > s, Z)),
\end{align}
and where $e_\textup{type}(\beta; s\given z) = \E(v_\textup{type}(\mathcal{X}; s; \beta) \given T > s, Z=z)$.
Additionally, under Assumption~\ref{ass:indep_cens},
\begin{equation} \label{eq:J_general}
J_{\textup{type}}(\beta) = \E\big( \frac{\partial}{\partial \beta\T} A(\beta; X) (Y- \mu(\beta; X)) - A(\beta; X) \frac{\partial}{\partial \beta \T} \mu(\beta; X) \big)
\end{equation}
for each of the three types, and so, $J_{\textup{type}}(\beta)$ does not depend on the type.
\end{lemma}

\begin{proof}
It can be seen that
\begin{equation}
W^* - 1 = - \int_0^{t-} \frac{1}{G(s \given Z)} M(\mathcal{\tilde X}; \d s \given Z)
\end{equation}
and 
\begin{equation}
e^*(s \given Z) = \frac{\E(Y \given T > s, Z)}{G(s \given Z)}
\end{equation}
under the assumption. This reveals how the structure concerning $u(\mathcal{X};\beta)$ and $v_\textup{type}(\mathcal{X};s;\beta)$ comes from $u_{\textup{type}}(\mathcal{\bar X}; \beta)$ under the assumption. The remaining part concerning $e_\textup{type}(\beta; s \given Z)$, on the other hand, comes from $\E(\dot u_{\textup{type}}(\mathcal{\bar X}_1; \mathcal{\bar X}; \beta) \given \mathcal{\bar X})$.
For types \textit{ind} and \textit{out}, the structure follows since, for any suitable function $f$ of the underlying information $\mathcal{X}$,
\begin{equation}
\begin{aligned}
& \phantom{{}=} \E\big( f(\mathcal{X}) W \frac{\indic{\tilde T > s} \indic{Z = z}}{\tilde S(s \given z) p_Z(z)} \big) \\
& = \E \big(f(\mathcal{X}) \frac{\indic{C \geq T \wedge t}}{G(T \wedge t - \given Z)} \given \tilde T > s, Z=z \big) \\
& = \E \big(f(\mathcal{X}) \frac{\P(C \geq T \wedge t \given C > s, Z=z, T)}{G(T \wedge t - \given Z)} \given \tilde T > s, Z=z \big) \\
&= \frac{\E(f(\mathcal{X}) \given \tilde T > s, Z=z)}{G(s \given z)} \\
&= \frac{\E(f(\mathcal{X}) \given T > s, Z=z)}{G(s \given z)},
\end{aligned}
\end{equation}
where the independent censoring assumption is used in the last two equalities.
To handle the \textit{pse} case, also note that $\E(M(\mathcal{\tilde X};s \given z) \given \mathcal{X}) = 0$ under the assumption, and thus $\E(\dot \Lambda(\mathcal{\tilde X};s \given z) \given \mathcal{X}) = 0$ as well.
This takes care of most of the terms from $\ddot \theta$ in $\dot u_\textit{pse}$. One remaining term, involving $A(\beta_0; X) W Y$, follows the structure from above. For the last remaining term, note that
\begin{equation}
\begin{aligned}
&\phantom{{}=} \frac{\E(A(\beta_0;X) e^*(s \given z) R(\mathcal{\tilde X};s;z))}{\check S(s \given z) p_Z(z)} \\
& = \frac{\E(A(\beta_0;X) e^*(s \given z) \indic{T > s} \indic{Z=z} \P(C \geq s \given \mathcal{X}))}{S(s \given z) G(s- \given z) p_Z(z)} \\
&= \E(A(\beta_0; X) e^*(s \given z) \given T > s, Z=z) \\
& = \frac{\E(A(\beta_0; X) \E(Y \given T > s, Z=z) \given T > s, Z=z)}{G(s \given z)}
\end{aligned}
\end{equation}
since $\P(C \geq s \given \mathcal{X}) = G(s- \given z)$ on $Z=z$ under the independent censoring assumption.
The structure of \eqref{eq:h_type_representation} and the fact that $\E(M(\mathcal{\tilde X};s \given z) \given \mathcal{X}) = 0$ reveals that 
\begin{equation}
J_{\textup{type}}(\beta) = \E(\frac{\partial}{\partial \beta} h_{\textup{type}}(\mathcal{\bar X}; \beta)) = \E(\frac{\partial}{\partial \beta} u(\mathcal{X}; \beta))
\end{equation}
which will have the desired expression and does not depend on the type of approach.
\end{proof}
It may also be noted that, in Lemma~\ref{lemma:rep_independent}, if the model $\E(Y \given X) = \mu(\beta_0;X)$ holds, then, at the true $\beta = \beta_0$, 
\begin{equation}
J_{\textup{type}}(\beta) = - \E\big(A(\beta; X) \frac{\partial}{\partial \beta} \mu(\beta; X) \big).
\end{equation}

The proofs of Theorem~\ref{thm:influence} and Theorem~\ref{thm:varest} are now presented.
\begin{proof}[Proof of Theorem~\ref{thm:influence}]
The true $\beta_0$ becomes the limiting solution and the result can be established by appealing to Lemma~\ref{lemma:estimator_exist} and Lemma~\ref{lemma:rep_independent}. Note how $M(\mathcal{\tilde X}; s \given z)$ is $\int_0^s \indic{T > u} M(\d u \given z)$ under the assumptions in the notation of the theorem. Under Assumption~\ref{ass:indep_cens}, and using Lemma~\ref{lemma:rep_independent}, this implies that $\beta_0$ is in fact solving $\E(h_{\textup{type}}(\bar{\mathcal{X}}; \beta_0)) = 0$ for each type since $\E(u(\mathcal{X};\beta_0)) = 0$ and $\E(M(\tilde{\mathcal{X}}; s \given Z) \given \mathcal{X}) = 0$ for all $s$.
Also, the matrix $J_{\textup{type}}(\beta_0)$ reduces to the relevant $J(\beta_0)$ under the model assumption $\E(Y \given X) = \mu(\beta_0; X)$ as just noted. Using Lemma~\ref{lemma:estimator_exist} now gives the desired result and also the expression of $\dot \beta_{\textup{type}}$ via \eqref{eq:h_type_representation}.
\end{proof}

\begin{proof}[Proof of Theorem~\ref{thm:varest}]
The $\frac{1}{n}\frac{\partial}{\partial \beta} U_{n,\textup{type}}(\beta)$, if evaluated at the true $\beta_0$, converge to $J(\beta_0)$ under the assumptions. 
An application of Lemma~\ref{lemma:contribution_approximation} will give a close approximation of $\frac{1}{n}\sum_{i=1}^n u_{i,n,\textup{type}}(\beta) u_{i,n,\textup{type}}(\beta)\T$ to a $U$-statistic of order 3 which will, according to the law of large numbers for $U$-statistics converge to its mean. That mean is $\E(u_{\textup{type}}(\mathcal{\bar X}; \beta)^{\otimes 2})$.
Following the approach of the proofs of Lemma~\ref{lemma:rep_independent} and Corollary~\ref{cor:as_variance_express} will lead to the desired expression at the true $\beta_0$.
Under the mentioned regularity conditions, which allow for the mentioned convergences to be uniform in an open neighborhood of the true $\beta_0$, evaluating at estimates that converge to the true $\beta_0$ will yield the same limit as at the true $\beta_0$.
The difference
\begin{equation}
    \Sigma_{\textup{type}}' - \Sigma_{\textup{type}} = \E(\int_0^{t-} \E(\phi_{\textup{type}}(s; T, D, X) \given T > s, Z=z)^{\otimes 2} \frac{S(s \given Z)}{G(s \given Z)} \Lambda(\d s \given Z))
\end{equation}
is non-negative definite necessarily and positive definite unless $\E(\phi_{\textup{type}}(s; T, D, X) \given T > s, Z=z) = 0$ for $\Lambda(\cdot \given z)$-almost all $s$ for almost all $z$ as claimed since $S(t \given Z) > 0$ almost surely by assumption.
\end{proof}

\end{document}